\documentclass[11pt]{article}
\usepackage[left=1in, top=1in, right=1in, bottom=1in]{geometry}
\usepackage{amsmath, amssymb, amsthm, amsfonts}
\usepackage{mathabx}
\usepackage[numbers]{natbib}
\usepackage{float}
\usepackage{caption}
\usepackage{xfrac}
\usepackage[T1]{fontenc}
\usepackage{diagbox}

\usepackage{booktabs} 
\usepackage[ruled]{algorithm2e} 
\usepackage{algorithmic}

\SetAlFnt{\small}
\SetAlCapFnt{\small}
\SetAlCapNameFnt{\small}
\SetAlCapHSkip{0pt}
\IncMargin{-\parindent}

\newtheorem{theorem}{Theorem}[section]

\newtheorem{lemma}[theorem]{Lemma}
\newtheorem{proposition}[theorem]{Proposition}

\newtheorem{definition}[theorem]{Definition}

\usepackage{colortbl}

	\definecolor{amethyst}{rgb}{0.6, 0.4, 0.8}

\newcommand{\be}{\begin{equation}}
\newcommand{\ee}{\end{equation}}
\newcommand{\beq}{\begin{equation*}}
\newcommand{\eeq}{\end{equation*}}

\newcommand{\argmax}{\mathop{\rm argmax}}

\newcommand{\Z}{\mathbb{Z}}

\newcommand{\I}{\mathbb{I}}

\newcommand{\eps}{\varepsilon}

\newcommand{\AutoAdjust}[3]{\mathchoice{ \left #1 #2  \right #3}{#1 #2 #3}{#1 #2 #3}{#1 #2 #3} }
\newcommand{\Xcomment}[1]{{}}

\newcommand{\InParentheses}[1]{\AutoAdjust{(}{#1}{)}}
\newcommand{\InBrackets}[1]{\AutoAdjust{[}{#1}{]}}
\newcommand{\Ex}[2][]{\operatorname{\mathbf E}_{#1}\InBrackets{#2}}
\newcommand{\Exlong}[2][]{\operatornamewithlimits{\mathbf E}\limits_{#1}\InBrackets{#2}}
\newcommand{\Prx}[2][]{\operatorname{\mathbf{Pr}}_{#1}\InBrackets{#2}}
\newcommand{\Prlong}[2][]{\operatornamewithlimits{\mathbf{Pr}}\limits_{#1}\InBrackets{#2}}

\newcommand{\eqdef}{\overset{\mathrm{def}}{=\mathrel{\mkern-3mu}=}}
\newcommand{\vect}[1]{\ensuremath{\mathbf{#1}}}

\newcommand{\iprod}[2]{\langle #1, #2 \rangle}
\newcommand\restr[2]{{
  \left.\kern-\nulldelimiterspace 
  #1 
  \vphantom{\big|} 
  \right|_{#2} 
  }}
\def\prob{\Prx}
\def\Prl{\Prlong}

\newcommand{\weights}{\vect{w}}
\newcommand{\wweights}{\widetilde{\vect{w}}}
\newcommand{\polymatch}{\mathcal{F}_{_\mathcal{M}}}

\newcommand{\alg}{\textsf{ALG}}
\newcommand{\opt}{\textsf{OPT}}
\newcommand{\fopt}{\textsf{f-OPT}}
\newcommand{\exopt}{\textsf{ex-ante-OPT}}

\newcommand{\witu}{u^*}
\newcommand{\witv}{v^*}

\newcommand{\realedges}{\widetilde{E}}

\newcommand{\realincv}[1][v]{\widetilde{E}_{#1}}
\newcommand{\realincu}[1][u]{\widetilde{E}_{#1}}
\newcommand{\mincv}[1][v]{E_{\text{-}#1}}
\newcommand{\realhmincv}[1][v]{\widehat{E}_{\text{-}#1}}

\renewcommand{\emptyset}{\varnothing}

\newcommand{\feasible}{\mathcal{M}}
\newcommand{\bfamily}{\mathcal{B}}
\newcommand{\bfamilyvertex}{\mathcal{B}^{v.a.}}
\newcommand{\bfamilyedge}{\mathcal{B}^{e.a.}}

\newcommand{\vr}{\vect{r}}
\newcommand{\vs}{\vect{s}}
\newcommand{\vy}{\vect{y}}
\newcommand{\vrt}[1][t]{\vect{r}^{#1}}
\newcommand{\vwt}[1][t]{\vect{w}^{#1}}

\newcommand{\dist}{\mathbf{F}}
\newcommand{\dists}{\vect{\dist}}

\newcommand{\disti}[1][i]{{F_{#1}}}

\newcommand{\marg}[1][e]{x_{#1}}
\newcommand{\margv}{\vect{x}}

\newcommand{\margmatch}[1][e]{y_{#1}}
\newcommand{\margmatchv}{\vect{y}}

\title{Online Stochastic Max-Weight Matching:\\ 
prophet inequality for vertex and edge arrival models\thanks{This work is supported by Science and Technology Innovation 2030 –“New Generation of Artificial Intelligence” Major Project No.(2018AAA0100903), Innovation Program of Shanghai Municipal Education Commission, Program for Innovative Research Team of Shanghai University of Finance and Economics (IRTSHUFE) and the Fundamental Research Funds for the Central Universities. The first two authors are partially supported by the European Research Council (ERC) under the European Union's Horizon 2020 research and innovation program (grant agreement No. 866132), and by the Israel Science Foundation (grant number 317/17).}}

\author{Tomer Ezra\thanks{Tel Aviv University. Email:  \texttt{tomer.ezra@mail.tau.ac.il}, \texttt{mfeldman@tau.ac.il}}
\and
Michal Feldman\footnotemark[2]
\and
Nick Gravin\thanks{ITCS, Shanghai University of Finance and Economics. Email: \texttt{\{nikolai, tang.zhihao\}@mail.shufe.edu.cn}} \and Zhihao Gavin Tang\footnotemark[3] 
}

\begin{document}

\maketitle

\begin{abstract}
We provide {\em prophet inequality} algorithms for online weighted matching in general (non-bipartite) graphs, under two well-studied arrival models, namely edge arrival and vertex arrival. The weight of each edge is drawn independently from an a-priori known probability distribution.
Under edge arrival, the weight of each edge is revealed upon arrival, and the algorithm decides whether to include it in the matching or not. 
Under vertex arrival, the weights of all edges from the newly arriving vertex to all previously arrived vertices are revealed, and the algorithm decides which of these edges, if any, to include in the matching.
To study these settings, we introduce a novel unified framework of {\em batched prophet inequalities} that captures online settings where elements arrive in batches; in particular it captures matching under the two aforementioned arrival models.
Our algorithms rely on the construction of suitable {\em online contention resolution scheme} (OCRS).
We first extend the framework of OCRS to batched-OCRS, we then establish a reduction from batched prophet inequality to batched OCRS, and finally we construct batched OCRSs with selectable ratios of $0.337$ and $0.5$ for edge and vertex arrival models, respectively. 
Both results improve the state of the art for the corresponding settings. For the vertex arrival, our result is tight. 
Interestingly, a pricing-based prophet inequality with comparable competitive ratios is unknown.

\end{abstract}

\section{Introduction}
\label{sec:intro}
Online matching is a central problem in the area of online algorithms, and is extensively used in economics to model rapidly appearing online markets. 
Some prominent applications include matching platforms for ride sharing, healthcare (e.g., kidney exchange programs), job search, dating, and internet advertising. Internet advertising in particular is a killer application that has spurred a lot of research in the algorithmic and the algorithmic game theory communities on what came to be known as the ``adword problem'', introduced by~\citet{MehtaSVV07}, and studied extensively since (see, e.g., \cite{DevanurJ12,BuchbinderJN07}). The goal in this problem is to match online-appearing users' impressions to relevant advertisers. For a comprehensive recent survey on online matching and ad allocations, see \citet{Mehta13}.

The adword problem alone has instigated many interesting variants and theoretical results, which inherent many features of the original online matching model in bipartite graphs with \emph{one-sided vertex arrivals} considered in the seminal paper of~\citet{KarpVV90}. In this model, one side of the bipartite graph is fixed, and vertices on the other side arrive online. Upon arrival of a vertex, its edges to vertices on the other side are revealed, and a matching decision (i.e., whether to match and to which vertex) should be made immediately and irrevocably. Since the introduction of the online matching problem, other arrival models in bipartite graphs have been studied, including edge arrival models, and two-sided vertex arrivals, both of which extend naturally beyond bipartite graphs. In particular,~\citet{GamlathKMSW19} recently showed that for the case of bipartite graphs with edge arrivals, no online algorithm performs better than the straightforward greedy algorithm, which is $\frac{1}{2}$-competitive. On the positive side, they provide a $\frac{1}{2}+\Omega(1)$-competitive online matching algorithm for the model of general vertex arrivals introduced by~\citet{WangW15}. 

The results of \cite{WangW15} and \cite{GamlathKMSW19} apply to the unweighted matching problem (where every edge either exists or not). However, many of the online matching problems induced by real-life applications are {\em weighted} in nature, where every edge is associated with a real-valued {\em weight}. There has been a huge recent interest in studying the weighted variant of the online matching problem using a {\em Bayesian} approach, rather than the traditional worst-case approach (see, e.g.,~\cite{DevanurJSW19} and references therein).

In the Bayesian variant of the problem, one assumes stochastic instances of the problem, and the goal is to provide guarantees on the expected value of the matching. Indeed, there is a clear practical motivation for adapting stochastic models, as many of the online platforms keep accumulating huge amounts of statistical data and are heavily using it in their online decision making. In this paper we study the problem of online stochastic weighted matching in general graphs, under the two classical arrival models, namely edge arrival and vertex arrival. We study this problem within the framework of {\em prophet inequality}.
 
Prophet inequality is an important line of work that is highly relevant to online stochastic algorithms and algorithmic mechanism design, which addresses practical incentive issues that arise in many online market applications, in particular the need of the algorithm to solicit private information from the selfish agents.
Prophet inequality appeared first as a fundamental result in optimal stopping theory \cite{krengelS77,krengel1978semiamarts}. 
\citet{HajiaghayiKS07} were the first to realize the applicability of the prophet inequality framework within mechanism design applications. 
Later, \citet{ChawlaHMS10} applied the prophet inequality framework to the design of sequential posted price mechanisms that give approximately optimal revenue in a Bayesian multi-parameter unit-demand setting (BMUD).

An important ingredient in the result of \cite{ChawlaHMS10} is the first constant competitive (specifically, $\frac{4}{27}$) prophet inequality for the online weighted matching problem with edge arrival in a bipartite graph. In this version of the prophet inequality, the weight of every edge is drawn independently from an a-priory known probability distribution, and the edges arrive online in an unknown order. Upon arrival of an edge, its weight is revealed to the algorithm, which should in turn decide, immediately and irrevocably, whether to include this edge in the matching.
The goal is to maximize the expected total weight of the selected matching, and its performance is compared with the expected total weight of the matching selected by a {\em prophet}, who knows the weights of all the edges in advance, and thus selects the maximum weight matching for every realized instance of the graph. The algorithm of~\cite{ChawlaHMS10} is non adaptive, meaning that the online algorithm calculates thresholds for all the edges before observing any weights, and accepts every edge if and only if (i) it can be feasibly added to the current matching, and (ii) the weight of the edge exceeds (or equal to) its threshold. 

Later, \citet{KleinbergW19} introduced a general combinatorial prophet inequality for a broad class of Bayesian selection problems, where the feasible set is represented as an intersection of $k$ matroids. They found an adaptive\footnote{The algorithm calculates the threshold 
at the time of the element's arrival. This threshold depends on the arrival order and the weights of all previously appearing elements.} $\frac{1}{4k-2}$-competitive algorithm for this setting. They also showed that it can be used for the design of a truthful mechanism in the BMUD setting with more general feasibility constraints than in \cite{ChawlaHMS10}. The most recent result for bipartite matching with edge arrivals is a $\frac{1}{3}$-competitive non-adaptive prophet inequality by~\citet{GravinW19}. They also gave an upper bound of $\frac{4}{9}$, which establishes a clear separation between weighted and unweighted graphs (indeed, in the unweighted case, a simple greedy algorithm is $\frac{1}{2}$-competitive, even in a non-stochastic setting). Some of the aforementioned results are summarized in Table~\ref{tab:results}.

\begin{table}[]
	\begin{tabular}{l l l l l}
		\cline{1-4}
		\multicolumn{1}{| l |}{}                                                                  & \multicolumn{1}{c|}{Edge Arrival}                                                                                                                                                                                                              & \multicolumn{2}{c|}{ Vertex Arrival}                                                                                 &  \\ 
\cline{3-4}		
		\multicolumn{1}{| l |}{}                                                                  & \multicolumn{1}{c|}{Bipartite$\left . \middle/\vphantom{\sum_{a}^b} \right .$Any graphs}                                                                                                                                                                                                              & \multicolumn{1}{c|}{\begin{tabular}[c]{@{}c@{}}One Sided\\ (Bipartite graphs)\end{tabular}} & \multicolumn{1}{c|}{
\begin{tabular}[c]{@{}c@{}}Two Sided\\ (Bipartite)\end{tabular}$\left . \middle/\vphantom{\sum_{a}^b} \right .$Any graphs}                                                                                                                                                            &  \\ 		
\cline{1-4}
\multicolumn{1}{|l|}{\begin{tabular}[c]{@{}l@{}}Worst-case\\ (unweighted)\end{tabular}} & \multicolumn{1}{l|}{\begin{tabular}[c]{@{}l@{}}$\frac{1}{2}$ (tight)\\ bipartite \& any graphs\\ \cite{GamlathKMSW19}\end{tabular}}                                                                                                                        & \multicolumn{1}{l|}{\begin{tabular}[c]{@{}l@{}}$1-\frac{1}{e}$ (tight)\\ \cite{KarpVV90}\\  \end{tabular}}                                              & \multicolumn{1}{l|}{\begin{tabular}[c]{@{}l@{}}$\geq \frac{1}{2}+\Omega(1)$(any graph)\\ \cite{GamlathKMSW19}\\ \\ $\leq 0.591$(bipartite graphs)\\ \cite{BuchbinderST17}\end{tabular}} &  \\ 
\cline{1-4}
		\multicolumn{1}{|l|}{\begin{tabular}[c]{@{}l@{}}Bayesian\\ (weighted)\end{tabular}}     & \multicolumn{1}{l|}{\begin{tabular}[c]{@{}l@{}}$\geq \frac{1}{3}$ (bipartite graphs)\\ $\le \frac{4}{9}$ (bipartite graphs)\\ \cite{GravinW19}\\  \\  \boldmath$\geq 0.337$ {\bf(any graph)}\\ (Theorem~\ref{thm:improved-ocrs})\end{tabular}} & \multicolumn{1}{l|}{\begin{tabular}[c]{@{}l@{}}$\frac{1}{2}$ (tight)\\ \cite{FeldmanGL15}\\ \end{tabular}}                                                    & \multicolumn{1}{l|}{\begin{tabular}[c]{@{}l@{}} \boldmath$\ge\frac{1}{2}$ {\bf (any graph)} \\  (Theorem~\ref{thm:ocrs-vertex}) \\ \\ $\le\frac{1}{2}$ (tight)\\ even for one sided arrival\\ bipartite graphs\end{tabular}}                   &  \\ \cline{1-4}
		&                        &                                                                                        && 
	\end{tabular}
\caption{Competitive ratios for online matching: previous and new results. Our new results are indicated in bold face.\label{tab:results}}
\end{table}

\paragraph{Batched Prophet Inequality}
The common assumption in the literature on prophet inequalities is that the elements arrive one by one, thus the online algorithm makes a simple binary decision at the arrival of a new element. A notable exception is the setting of Bayesian combinatorial auctions considered by \citet{FeldmanGL15}.
In this scenario, multiple items should be allocated to $m$ selfish buyers, each of which has a potentially complex combinatorial (submodular) valuation over different sets of items. The valuation function of each buyer is drawn from an arbitrary a-priori known probability distribution, which are mutually independent across the $m$ buyers. The goal is to maximize the social welfare, i.e., the total sum of buyer values for their allocated sets, by a truthful mechanism. The proposed solution uses \emph{posted pricing} and follows a similar charging argument as~\citet{ChawlaHMS10} and~\citet{KleinbergW19}, by decomposing the welfare into revenue and surplus. 

Unlike the previous results on prophet inequality, the online algorithm in~\cite{FeldmanGL15}: (i) observes multiple elements at each time step, and thus should make a complex allocation decision (as opposed to a binary decision whether or not to accept the element), and (ii) the values of basic elements (allocation of item $j$ to buyer $i$) might be dependent across different items. To highlight the similarities and differences with other work on prophet inequalities, it is informative to consider the restriction of the combinatorial valuations in~\cite{FeldmanGL15} to the class of unit-demand buyers. In this case the setting and the proposed solution in~\cite{FeldmanGL15} correspond to the one-sided vertex arrival model of~\cite{KarpVV90},  
where elements (weighted edges of a bipartite graph) are revealed online in \emph{batches}, and the online algorithm can choose one element out of many. On the other hand, the prior distribution, unlike~\cite{KleinbergW19} or ~\cite{GravinW19}, does not need to be independent: a buyer's values for different individual items may be dependent.
 
In this work, we propose a new framework, termed ``\emph{Batched prophet inequality}'', that captures both the edge arrival model from~\cite{KleinbergW19,GravinW19} and the one-sided vertex arrival model from~\cite{FeldmanGL15}. Within this framework, we generalize the model of bipartite matching with one-sided vertex arrivals to two-sided vertex arrivals, and also extend our results for both the vertex and edge arrivals to general (non bipartite) graphs.
Unlike~\cite{KleinbergW19,GravinW19,FeldmanGL15}, who take a ``pricing/charging'' approach, our solution relies on two novel Online Contention Resolution Schemes (OCRS) for batched arrival settings.  

Contention Resolution Schemes (CRS) were  introduced by~\citet{ChekuriVZ14} as a powerful rounding technique in the context of submodular maximization. The CRS framework was extended to the OCRS framework for online stochastic selection problems by~\citet{FeldmanSZ16}, who provided OCRSs for different problems, including intersections of matroids and matchings. One particularly important application of OCRS was Prophet Inequality~\cite{FeldmanSZ16,LeeS18}. Specifically, for matching feasibility constraint, \citet{FeldmanSZ16} constructed a $\frac{1}{2e}$-OCRS that implies $\frac{1}{2e}$-competitive algorithm for the prophet inequality setting of matching with edge arrival.

\subsection{Our Results}
\label{sec:results}
We provide prophet inequalities for matching feasibility constraint in \emph{weighted general (non-bipartite)} graphs in the two fundamental online models of {\em vertex arrivals} and {\em edge arrivals}.
To obtain our results for these two different settings, we use a common approach that relies on two similar Online Contention Resolution Schemes (OCRS).
To this end, we introduce a novel unified framework of \emph{Batched-OCRS} that enables the analysis of settings where at each time step multiple elements arrive together as a batch, and a complex online decision should be made. This framework captures (among other models) both the vertex arrival and the edge arrival models in online matching. 

\paragraph{Reduction from prophet inequality to OCRS.}
Our first result is a general reduction from Batched-Prophet Inequality to Batched-OCRS for any downward-closed feasibility constraint.
This general reduction implies that to get prophet inequalities with a certain competitive ratio, it suffices to construct an OCRS with the same selectable ratio.
We then construct batched-OCRSs for both vertex and edge arrival models. 


\paragraph{Matching with vertex arrival.}
We present a simple $\frac{1}{2}$-batched OCRS, and thus a $\frac{1}{2}$-competitive prophet inequality algorithm for general (non-bipartite) graphs.
This result is tight (a matching upper bound is derived from the classical prophet inequality problem). 
Our competitive ratio holds also with respect to the stronger benchmark of the optimal fractional matching. Unlike bipartite graphs, the optimal fractional matching in general graphs may indeed have a strictly higher weight than any integral matching. 

An interesting implication of this result is that in the Bayesian setting there is no gap between the competitive ratio that can be obtained under the 1-sided and 2-sided vertex arrival models in bipartite graphs, or even under vertex arrival in general graphs. This is in contrast to the non-Bayesian (worst case) online model, where there is a gap between 1-sided and 2-sided vertex arrivals (see Table~\ref{tab:results}).

Our vertex arrival model restricted to the case of bipartite graphs is an importation of the two-sided vertex arrival model~\cite{WangW15, GamlathKMSW19} from the online matching literature to the online Bayesian selection problem. It generalizes the setting of~\cite{FeldmanGL15} for unit-demand buyers to the model where \emph{buyers and items} can both appear online. 

\paragraph{Matching with edge arrival.}
We revisit the prophet inequality matching problem with edge arrivals considered by~\citet{GravinW19} and~\citet{KleinbergW19}. While \cite{GravinW19} and \cite{KleinbergW19} took a pricing/threshold-based approach, we use the OCRS approach. 
We first show that a simple OCRS already gives a $\frac{1}{3}$-competitive algorithm, matching the previous best bound of~\cite{GravinW19} for bipartite graphs, and generalizing their result to general non-bipartite graphs\footnote{Note, however, that unlike~\cite{GravinW19}, the OCRS-based algorithm is adaptive.}. 
We further improve the competitive ratio to $0.337$ by constructing a better OCRS, which requires more subtle analysis. 
These results hold against the even stronger benchmark of the {\em ex-ante} optimal solution that satisfies fractional matching constraints (similar to the observation in~\citet{LeeS18}).

\subsection{Our Techniques}
\label{sec:techniques}
As mentioned earlier, the OCRS approach is not as common as the pricing approach in prophet inequality settings. We note that the earlier algorithms of~\citet{ChawlaHMS10} and \citet{Alaei11}, when applied to the classic prophet inequality setting, become a simple $\frac{1}{2}$-competitive algorithm that is indeed a $\frac{1}{2}$-OCRS. These algorithms also appear to be closer in spirit to our OCRS approach than to the more recent papers on prophet inequality (e.g., \cite{KleinbergW19,FeldmanGL15,GravinW19,DuettingFKL17,Lucier17,EdenFFS2018}). 

One of the reasons that OCRSs are not as prevalent in prophet inequality settings is that the formal definition of OCRSs is not specifically tailored for prophet inequalities. As a result, the approximation factors that are obtained by OCRSs are not as tight.
For example, the original OCRS introduced by~\citet{FeldmanSZ16} for matching feasibility constraint achieves a competitive ratio of $\frac{1}{2e}$, whereas even a non-adaptive pricing-based algorithm achieves the much better ratio of $\frac{1}{3}$ (\citet{GravinW19}). 

Indeed, these OCRSs are usually designed to work against a strong almighty adversary, who controls the arrival order of the elements and knows in advance the realization of the instance and the random bits of the algorithm. 
The OCRSs we construct in this work are better tailored to the prophet inequality setting as they are designed against a weaker oblivious adversary, who can select an arbitrary order of element arrivals, but does not observe the algorithm's decisions and the realization of the instance. 

Our $\frac{1}{3}$- and $\frac{1}{2}$-selectable OCRSs for respective edge and vertex arrival models are surprisingly simple and intuitive. Moreover, the latter OCRS already gives a tight result for the vertex arrival model. We note, however, that formulating a general OCRS framework for batched arrivals of elements is not as trivial as it might seem at a first glance. We give an example in Appendix~\ref{sec:bad_batched_OCRS} illustrating why a simpler and apparently more natural than ours extension of OCRS to the setting with batched arrivals can be problematic. 

Even more surprisingly, at the time of writing this paper, no pricing-based approach is known to match the $\frac{1}{2}$-competitive guarantee attainable by the OCRS for the vertex arrival model. Moreover, several natural attempts of generalizing the pricing scheme in \citet{FeldmanGL15} fail miserably, even for bipartite graphs. For example, one natural generalization would be to set the price on a new vertex $v$ to be half of the expected contribution of the future edges incident to $v$ to the optimum matching. As it turns out, this pricing scheme achieves a competitive ratio as small as $\frac{1}{4}$. This is demonstrated in Appendix~\ref{sec:pricing_algorithm}.

On the other hand, the $\frac{1}{3}$-OCRS for the edge arrival model is based on a simple union bound which still leaves some room for improvement. We improve the ratio of  $\frac{1}{3}$ to $0.337$ by bounding the negative correlation for any pair of events that vertex $u$ and vertex $v$ are matched at the time of edge $(uv)$ arrival. This turned out to be most technical part of the paper.


\subsection{Related Work}
There is an extensive literature regarding online matching and stochastic matching problems. Below we survey the studies that are most related to our work. Recently, the ``fully online matching" model has been studied~\citet{HuangKTWZZ18, AshlagiBDJSS19, HuangPTTWZ19}, motivated by ride sharing applications. This is a different vertex arrival model in which all vertices from a general graph arrive and depart online. It is possible to study the stochastic/prophet inequality version of the fully online model, which we leave as an interesting future direction.

\citet{GravinTW19} studied the online stochastic matching problem with edge arrivals (a.k.a. the unweighted version of the prophet inequality with edge arrivals in this paper) and achieved a $0.506$-competitive algorithm. The stochastic matching setting is also studied in the (offline) query-commit framework. The input of this problem is an (unweighted) graph associated with the existence probabilities of all edges.
The algorithm can query the existence of the edges in any order. However, if an edge exists, it has to be included into the solution. The Ranking algorithm by~\citet{KarpVV90} induces an $(1-\frac{1}{e})$-competitive algorithm for this problem on bipartite graphs. \citet{CostelloTT12} provided a $0.573$- competitive algorithm on general graphs and proved a hardness of $0.898$. \citet{GamlathKS19} provided a $1-\frac{1}{e}$-competitive algorithm for the weighted version of this problem.

Online contention resolution schemes have also been studied in settings beyond worst case arrivals. \citet{Adamczyk018} considered the random order model and constructed $\frac{1}{k+1}$-OCRS for intersections of $k$ matroids. \citet{LeeS18} constructed optimal $\frac{1}{2}$-OCRS and $(1-\frac{1}{e})$-OCRS for matroids with arbitrary order and random order, respectively. Offline contention resolution schemes for matching have also attracted attention due to its applications in submodular maximization problems~\cite{ChekuriVZ14, FeldmanNS11, BruggmannZ19}, and the connection between the correlation gap and contention resolution schemes~\cite{GuruganeshL17}. We refer the interested readers to \citet{BruggmannZ19} for a comprehensive recent survey on the topic.

Since prophet inequality problems were first introduced \citet{krengelS77,krengel1978semiamarts,samuel1984comparison}, many variants have been developed over the years. A recent line of work has considered sample based variants, where the distributions of the values are not given explicitly, and the challenge is to provide good competitive ratios using a limited number of samples  \citet{azar2014prophet,correa2019prophet,correa2020two,ezra2018prophets,rubinstein2019optimal}.
Another related line of work, initiated in \citet{kennedy1985optimal,kennedy1987prophet,kertz1986comparison}, has considered multiple-choice prophet inequalities, and was later extended to combinatorial settings such as matroid (and matroids intersection) \citet{KleinbergW19,azar2014prophet}, polymatroids \citet{dutting2015polymatroid}, and general downward closed feasibility constrains \citet{rubinstein2016beyond}. 

\subsection{Paper Roadmap} \citet{FeldmanSZ16} define the notion of online contention resolution scheme (OCRS) to study settings where elements arrive online, and establish a reduction from prophet inequality to OCRS. 
In Section \ref{sec:prelim} we extend the OCRS and prophet inequality frameworks to settings where elements arrive online in batches. We begin by introducing the general setting of batched arrival. In Section~\ref{sec:batched-ocrs} we extend the notion of OCRS to batched-OCRS.
In Section~\ref{sec:batched-pi} we extend the notion of prophet inequality to batched prophet inequality. 
In Section~\ref{sec:reduction-pi-ocrs} we establish a reduction from batched prophet inequality to batched OCRS.
In Section~\ref{sec:matching} we present a natural special case of batched prophet inequality, namely graph matching prophet inequality.  Then, in Sections ~\ref{sec:vertex-arrival} and~\ref{sec:edge-arrival} we construct   OCRSs for vertex and edge arrival models, respectively. Finally, upper bounds on the competitive ratios for the prophet inequality with edge arrivals are provided in Section~\ref{sec:edge_lower}. Section~\ref{sec:future} concludes this paper with a list of open problems and future directions.

\section{Model and Preliminaries}
\label{sec:prelim}
Let $E$ be a set of elements, and let $\feasible$ be a downward closed family of feasible subsets of $E$, i.e., if $S\in\feasible$, then $S'\in\feasible$ for any $S'\subseteq S$. 
The elements in $E$ are partitioned into $T$ disjoint sets (batches) $B_1, \ldots, B_{T}$ that arrive online in the order from batch $B_1$ to batch $B_T$. I.e., at time $t$, all elements of batch $B_t$ appear simultaneously. 
The partition of elements into the batches and their arrival order $\InParentheses{B_t}_{t\in[T]}$ should conform to a certain structure formally specified by a family of all feasible ordered partitions $\bfamily$ of $E$. 

Some examples of feasible ordered partitions include the following: (i) all batches in $\bfamily$ are required to be singletons, (ii) given a partial order $\pi$ on $E$, a feasible ordered partition $\InParentheses{B_t}_{t\in[T]}\in\bfamily$ is required to have $\pi(e_t)\le\pi(e_s)$ for any $e_t\in B_t, e_s\in B_s$ where $t<s$, (iii) suppose the set of elements $E$ consists of the edges of a bipartite graph $G=(L,R;E)$, and each batch $B_t$ must contain all edges incident to a vertex $u\in L$.

\subsection{Batched OCRS}
\label{sec:batched-ocrs}
For a given family of feasible batches $\bfamily$, 
consider a {\em sampling scheme} that selects a random subset $R\subseteq E$ as follows:
at time $t$, all elements of batch $B_t$ arrive, of which a random subset $R_t \subseteq B_t$ is {\em realized}. 
The realized sets $R_1, \ldots, R_{T}$ are mutually independent.
$R$ is then defined as the random set $R\eqdef \bigsqcup_{t\in [T]}R_t$. 

\citet{FeldmanSZ16} introduce the notion of $c$-selectable {\em online contention resolution scheme} (OCRS), as an online selection process that selects a feasible subset of $E$ such that every realized element $e \in R$ is selected with probability at least $c$, for the special case where every batch is a singleton.
We extend the definition of \citet{FeldmanSZ16} to batched OCRSs as follows.


\begin{definition}[$c$-selectable batched OCRS]
An online selection algorithm  $\alg$ is a batched OCRS with respect to a sampling scheme $R$ if it selects a set $I_t \subseteq R_t$ at every time $t$ such that $I \eqdef \bigsqcup_{t\in [T]} I_t$ is feasible (i.e., $I \in \feasible$).
It is called a $c$-selectable batched OCRS (or in short $c$-batched-OCRS) if:
\begin{equation}
\Prlong[I]{ e \in I_t ~\big\vert~ R_t=S} \ge c  \quad \text{ for all } t \in [T], S\subseteq B_t,   \mbox{ and } e\in S.
\label{eq:cocrs_int}
\end{equation}
	\label{def:batched_ocrs}
\end{definition}
The algorithm $\alg$ does not know the complete partition into batches and the arrival order of future batches. It only knows the general structure $\bfamily$.
Thus, at time $t$, $\alg$ chooses $I_t$  based on $B_1, \ldots, B_t$, and $R_1, \ldots, R_t$.

\subsection{Batched Prophet Inequality} 
\label{sec:batched-pi}
In batched prophet inequality, every element $e\in E$ has a weight $w_e$. Let $\vwt\eqdef\InParentheses{w_e}_{e \in B_t}$, and $\weights\eqdef\InParentheses{\vwt}_{t\in [T]}$. Weights are unknown a-priori, but for every $t$, $\vwt$ is independently drawn from a known (possibly correlated) distribution $\disti[t]$, and $\weights \sim \dists \eqdef\prod_{t\in [T]}\disti[t]$; I.e., we allow dependency within batches, but not across batches. Let $\weights(S) =\sum_{e \in S} w_e$ for any set $S\subseteq E$. As standard, let $\dists_{-t}= \prod_{i \neq t} \disti[i]$.
The particular partition of elements into batches and their order are a-priori unknown\footnote{
Note that $\dists$ might impose some constraints on the partition into batches: elements whose weights are dependent must belong to the same batch. No constraint is imposed on elements whose weights are independent and on the order of batches.}, except, of course, that $\InParentheses{B_t}_{t\in[T]}$ must conform to the general structure of $\InParentheses{B_t}_{t\in[T]}\in\bfamily$.
All elements of a batch $B_t$ and their weights $\vwt$ are revealed to the algorithm at time $t$.
We assume that the arrival order of the batches is decided by an oblivious adversary, i.e., the adversary can select an arbitrary  partition and order of arrival of the batches in $\bfamily$, but does not see the realization of the weights $w_e$ and the algorithm's decisions\footnote{The oblivious adversary is a standard assumption in the literature on online algorithms in stochastic settings.}. 
Let $\opt$ be a function that given weights $\weights$ returns a feasible set of maximum weight (i.e., $\opt(\weights) \in \arg\max_{S\in \feasible} \weights(S)$)\footnote{We assume that $\opt$ is deterministic (if a given weight vector $\weights$ induces multiple feasible sets of maximal weight, $\opt(\weights)$ returns one of them consistently).}.

\begin{definition} [$c$-batched-prophet inequality]
\label{def:prophet}
A batched-prophet inequality algorithm $\alg$ is an online selection process that selects at time $t$ a set  $I_t \subseteq B_t$ such that $I \eqdef \bigsqcup_{t\in [T]} I_t$ is feasible (i.e., $I \in \feasible$).
We say that $\alg$ has competitive ratio $c$ if
\begin{equation*}
\Exlong[\weights,I]{\weights(I)} \geq c \cdot \Exlong[\weights]{\weights(\opt(\weights))}.
\end{equation*}
\end{definition}



\subsection{Reduction: Prophet Inequality to OCRS} 
\label{sec:reduction-pi-ocrs}
We define a random sampling scheme $R(\weights,\dists)$ for $\weights\sim\dists$ as follows.
Let $R_t(\vwt,\dists)\eqdef B_t \cap \opt(\vwt,\wweights^{(t)})$ be the random subset of $B_t$ where
$\wweights^{(t)} \sim \dists_{-t}$ are generated independently of $\weights$, and $R(\weights,\dists)\eqdef \bigsqcup_{t \in [T]} R_t(\vwt,\dists)$. 
Note that:
 
\begin{enumerate}
\itemsep0em
	\item 
	The distribution of $R(\weights,\dists)$ is a product distribution over the random variables $R_t(\vwt,\dists)$.
	\item Since $\dists$ is a product distribution,  $(\vwt,\wweights^{(t)})\sim\dists$.
	\item  $\forall t \in [T]$, $R_t(\vwt,\dists)=R(\weights,\dists)\cap B_t$ has the same distribution as $\opt(\weights)\cap B_t$, where $\weights \sim \dists$.
	\item $\forall t \in [T]$,  $R_t(\vwt,\dists)\in\feasible$, and $\opt(\vwt,\wweights^{(t)}) \in \feasible$. But, $R(\weights,\dists)$ might not belong to $\feasible$.\label{property4}
	\item For every $t\in [T]$,
	\begin{equation}
		\Exlong[\weights, R]{\weights(R_t(\vwt,\dists))} = \Exlong[\weights]{\weights(\opt(\weights) \cap B_t)}. 		\label{eq:expected_marginal_int}
	\end{equation}
\end{enumerate}
 
\begin{theorem} [reduction from $c$-batched prophet inequality to $c$-batched OCRS]
	\label{thm:reduction_integral}
For every set $\bfamily$ of feasible ordered partitions, given a $c$-batched OCRS for the sampling scheme $R(\weights,\dists)$ with $\weights\sim\dists$, there is a batched prophet inequality algorithm for $\weights\sim\dists$ with competitive ratio $c$.
\end{theorem}
\begin{proof}	Consider the following online algorithm:	
	
\begin{algorithm}
\caption{Reduction from $c$-batched prophet inequality to $c$-batched OCRS}
\label{alg:integral}
\begin{algorithmic}[1]
	\FOR{$t\in\{1,...,T\}$} 
	\STATE Let $\vwt$ be the weights of elements in $B_t$
	\STATE Resample the weights $\wweights^{(t)} \sim \dists_{-t}$ 
	\STATE Let $R_t\gets \opt(\vwt,\wweights^{(t)}) \cap B_t$ \label{step:r}
	\STATE $I_t \gets$  $c$-OCRS($B_1,\ldots , B_t,R_1,\ldots,R_t$)\quad\quad for the structure $\bfamily$ of batches. 
	\ENDFOR	
	\STATE Return $I=\bigsqcup_{t \in [T]} I_t$
	\end{algorithmic}
\end{algorithm} 
Let $I$ be the random set returned by Algorithm \ref{alg:integral}, and $R_t$ (and resp. $R=\bigsqcup_{t\in[T]} R_t$) be the sets defined in step \ref{step:r} of the algorithm. It holds that
$$
\Exlong[\weights,R,I]{\weights(I)}  =   
\sum_{t \in[T]}  \Exlong[\weights,R,I]{\weights(I_t)} 
 =   \sum_{t \in[T]} \sum_{S\subseteq B_t} \Exlong[\weights,R,I]{\weights(I_t)~\Big \vert~ R_t=S} \Prl[\weights,R]{R_t=S}.
$$
Since $I_t$ and $\weights$ are independent given that $R_t=S$, we also have. 
\begin{eqnarray*}
	\Exlong[\weights,R,I]{\weights(I)} 
	 & = & \sum_{t \in[T]} \sum_{S\subseteq B_t} \Exlong[\weights,I]{\sum_{e \in B_t}w_e \cdot \Prl[I]{e \in I_t}~\bigg\vert~ R_t=S} \Prl[\vwt,R_t]{R_t=S}\\
	  & \stackrel{\eqref{eq:cocrs_int}}{\ge} & \sum_{t \in[T]} \sum_{S\subseteq B_t} \Exlong[\vwt]{\sum_{e \in S}w_e \cdot  c ~\bigg\vert~ R_t=S} \Prl[\vwt,R_t]{R_t=S}\\
	 & = & c \sum_{t \in[T]} \sum_{S\subseteq B_t} \Exlong[\weights]{\weights(S) ~\big\vert~ R_t=S} \Prl[\weights,R_t]{R_t=S} \\ 
	 & = & c \sum_{t \in[T]}  \Exlong[\weights,R]{\weights(R_t) } 
	  \stackrel{\eqref{eq:expected_marginal_int}}{=}  c \sum_{t \in[T]}  \Exlong[\weights]{\weights(\opt(\weights) \cap B_t)} 
	  =  c \cdot \Exlong[\weights]{\weights(\opt(\weights))}\qedhere
\end{eqnarray*}
\end{proof}

\section{Graph Matching}
\label{sec:matching}
An interesting special case of prophet inequality is the problem of selecting a matching in a graph.
Given a graph $G=(V,E)$ (not necessarily bipartite), the elements of the prophet inequality setting are the edges $e\in E$, and the family of feasible sets $\feasible$ is given by all matchings in $G$, i.e., $M\subseteq E$ is feasible $M\in\feasible$ iff  $e_1\cap e_2=\emptyset$ for any $e_1,e_2\in M$. 

We consider two different online arrivals models: (i) \emph{vertex arrival} and (ii) \emph{edge arrival}, which are natural special cases of our general framework of batched prophet inequality. 
%

\paragraph{Vertex arrival model.}
In the vertex arrival model, the vertices arrive in an arbitrary unknown order $\sigma$: $v_{\sigma(1)},\ldots,v_{\sigma(n)}$, where $v_{\sigma(i)}$ is the vertex arriving at time $i$. Upon arrival of vertex $v_{\sigma(i)}$, the weights on the edges from $v_{\sigma(i)}$ to all previous vertices $v_{\sigma(j)}$, where $j<i$, are revealed to the algorithm. The online algorithm must make an immediate and irrevocable decision whether to match $v_{\sigma(i)}$ to some available vertex $v_{\sigma(j)}$ such that $j<i$ (in which case $v_{\sigma(i)}$ and $v_{\sigma(j)}$ become unavailable), or leave $v_{\sigma(i)}$ unmatched (in which case $v_{\sigma(i)}$ remains available for future matches). 
Let $B_i^{\sigma} \eqdef \{(v_{\sigma(i)}v_{\sigma(j)}) ~\vert~ j<i \}$. The set of feasible ordered partitions for the vertex arrival model is
$$
\bfamilyvertex \eqdef \{(B_1^{\sigma},\ldots,B_{|V|}^{\sigma})\}_{\sigma \in S_V},
$$
where $S_V$ is the set of permutations over $V$.


\paragraph{Edge arrival model.}
In the edge arrival model, the edges arrive in an arbitrary unknown order $\sigma$: $e_{\sigma(1)}, \ldots, e_{\sigma(|E|)}$.
Upon arrival of edge $e=(uv)$, the algorithm must decide whether to match it (provided that $u$ and $v$ are still unmatched), or leave $e$ unmatched potentially saving $u$ and/or $v$ for future matches.
Let $B_i^{\sigma}$ be the singleton $\{e_{\sigma(i)}\}$. 
The set of feasible ordered partitions for the edge arrival model is
$$
\bfamilyedge \eqdef \{(B_1^{\sigma},\ldots,B_{|E|}^{\sigma})\}_{\sigma \in S_E},
$$
where $S_E$ is the set of permutations over $E$. 

\vspace{0.2cm}


Another extreme case is where all edges arrive in a single batch $B_1= E$. Then the online algorithm is no different than the offline algorithm, which can select the optimal offline solution $\opt(\weights)$. 

For the special case of the family of matching feasible sets, one can also consider {\em fractional matchings} $\margmatchv=(\margmatch)_{e\in E}$ specified by the matching polytope 
$$
\polymatch\eqdef\{\margmatchv ~\vert~ \forall v\in V~~ \sum_{u\in V}\margmatch[(uv)]\le 1,~\forall e\in E~~ \margmatch\ge 0\}.
$$ 
Every feasible ordered partition (that belongs to either $\bfamilyvertex$ or $\bfamilyedge$) induces a random variable $R(\weights,\dists)$ as defined in Section~\ref{sec:reduction-pi-ocrs} with respect to the set of feasible matchings $\feasible$.
Let $\marg[uv]^{opt}$ be the probability that $(uv) \in R(\weights,\dists)$, where $\weights \sim \dists$, and let $\margv^{opt} = (\marg[uv]^{opt})_{(uv) \in E}$.
Note that for any edge $(uv)$, $\marg[uv]^{opt}$ is precisely the probability that the edge $(uv)$ is in $\opt(\weights)$, where $\weights \sim \dists$. Therefore, $\margv^{opt} \in \polymatch$. 
Furthermore, by Property \ref{property4} in Section \ref{sec:reduction-pi-ocrs}, it holds that $|R_t(\vwt,\dists)| \leq 1$. 

In Section \ref{sec:vertex-arrival} we construct a $1/2$-batched OCRS for vertex arrival with respect to every $R$ whose corresponding $\margv$ belongs to $\polymatch$. Since $\margv^{opt} \in \polymatch$, we get a batched prophet inequality with competitive ratio $1/2$ for vertex arrival. Similarly, the $0.337$-batched OCRS for edge arrival in Section \ref{sec:edge-arrival} implies a batched prophet inequality with competitive ratio $0.337$ for edge arrival. 

\paragraph{Guarantees against stronger benchmarks.}
The guarantee in Definition~\ref{def:prophet} can be strengthen to hold against the stronger benchmark of the optimal fractional matching.
This extensions of the definition of batched OCRS and the reduction from batched prophet inequality to batched OCRS are deferred to Appendix~\ref{sec:fractional-ocrs}.
The construction of the OCRS for the setting of fractional matching in the vertex arrival model is deferred to Appendix~\ref{sec:vertex-arrival-fractional}.
In Appendix \ref{app:implications} we show  that for the edge arrival model, our construction actually gives an approximation to an even stronger benchmark, known as the ex-ante relaxation.

\section{A $1/2$-Batched OCRS for Matching with Vertex Arrival}
\label{sec:vertex-arrival}

In this section we construct a $\frac{1}{2}$-batched OCRS for the vertex arrivals. By the reduction in Theorem~\ref{thm:reduction_integral}, the constructed batched OCRS gives a batched prophet inequality with competitive ratio $1/2$ with respect to the optimal matching.

For every vertex $v$, let $R_v$ be an independent random subset of $B_v^{\sigma}$ generated by the sampling scheme, and let $R = \sqcup_v R_v$. For every edge $(uv)\in E$, let $\marg[uv] \eqdef \prob{(uv) \in R}$.
We write $u<v$ if vertex $u$ arrives before vertex $v$ in the vertex arrival order $\sigma$. 
\begin{theorem}
	If $R$  satisfies the following two conditions:
	\begin{equation}
	\sum_{u}\marg[uv]  \leq 1 \quad  \mbox{for every $v \in V$} \label{eq:sum_marg_integral}
	\end{equation}
	\begin{equation}
	|R_v| \leq 1 \quad \mbox{for every $v \in V$ and every realization of } R_v \label{eq:sumruv_integral}
	\end{equation} 
	Then, $R$ admits a $1/2$-batched OCRS for the $\bfamilyvertex$ structure of batches.
\label{thm:ocrs-vertex}
\end{theorem}

Note that $R$ defined in Section~\ref{sec:reduction-pi-ocrs} in the specific case of matching with vertex arrivals (see Section~\ref{sec:matching}) satisfies Equations \eqref{eq:sum_marg_integral} and \eqref{eq:sumruv_integral}.

\begin{proof}
	Upon the arrival of a vertex $v$, we compute  $\alpha_{u}(v)$ for every $u<v$ as follows:
\begin{equation}
\alpha_{u}(v) \eqdef \frac{1}{2 - \sum_{z<v}\marg[uz] }\leq
\frac{1}{2 - \sum_{z}\marg[uz] }
 \stackrel{\eqref{eq:sum_marg_integral}}{\leq} 1.
\label{eq:alpha_uv_vertex_integral}
\end{equation}
Note that $\alpha_u(v)$ cannot be calculated before the arrival of $v$.
We claim that the following algorithm is a $\frac{1}{2}$-batched OCRS with respect to $R$: 

\begin{algorithm}
	\caption{$1/2$-batched OCRS for vertex arrival}
	\label{alg:vertex-integral}
	\begin{algorithmic}[1]
		\FOR{$v\in\{1,...,|V|\}$} 
		\STATE Calculate $\marg[uz]=\prob{(uz) \in R}$ for all $u,z<v$ and $\alpha_u(v)$ for all $u<v$.
		\STATE Match the edge $(uv) \in R_v$ (if $R_v\neq\emptyset$) with probability $\alpha_{u}(v)$ if $u$ is unmatched.
		\ENDFOR
	\end{algorithmic}
\end{algorithm} 

Note that Algorithm~\ref{alg:vertex-integral} is well defined, since by Equation \eqref{eq:alpha_uv_vertex_integral}, $\alpha_{u}(v) \leq 1$ and Algorithm~\ref{alg:vertex-integral} matches no more than one vertex to $v$ by \eqref{eq:sumruv_integral}. It remains to show that Algorithm \ref{alg:vertex-integral} is a $1/2$-batched OCRS with respect to $R$.
We fix the arrival order $\sigma$. We prove by induction (on the number of vertices $|V|$) that  $\prob{(uv) \mbox{ is matched}}=\frac{\marg[uv]}{2}$. The base of the induction for $|V|=0$ is trivially true. To complete the step of the induction, we assume that 
$\prob{(uz) \mbox{ is matched}}=\frac{\marg[uz]}{2}$ for all $u,z<v$ and will show that $\prob{(uv) \mbox{ is matched}}=\frac{\marg[uv]}{2}$ for all $u<v$. In what follows, we say that ``$u$ is unmatched {\em at $v$}" if $u$ is unmatched {\em right before} $v$ arrives. 
 
\begin{equation}
\prob{u \mbox{ is unmatched at }v} = 1-\sum_{z<v}\prob{(uz) \mbox{ is matched}} = 1-\frac{1}{2}\sum_{z<v}\marg[uz],  
\label{eq:u_unmatched_integral} 
\end{equation}
where the second equality follows from the induction hypothesis. Therefore,
\begin{eqnarray*}
\label{eq:pr-e-matched_vertices_integral}
\prob{(uv) \mbox{ is matched}} & =  & \prob{u \mbox{ is unmatched at }v} \cdot \prob{(uv) \in R_v} \cdot \alpha_{u}(v) \\ 
  &\stackrel{\eqref{eq:alpha_uv_vertex_integral},\eqref{eq:u_unmatched_integral}}{=} & \left(1-\frac{1}{2}\sum_{z<v}\marg[uz]\right) \cdot\frac{1}{2 - \sum_{z<v}\marg[uz] } \cdot \marg[uv] =  \frac{\marg[uv]}{2}.\quad\quad\quad
\end{eqnarray*}


In order to prove that Algorithm \ref{alg:vertex-integral} is a $\frac{1}{2}$-batched OCRS with respect to $R$, we need to show 
that  $\prob{(uv) \in I_v \mid R_v =\{(uv)\}} \ge 1/2$ for every $u<v$. Indeed,
\begin{eqnarray*}
\prob{(uv) \in I_v \mid R_v =\{(uv)\}} & = &  \prob{u \mbox{ is unmatched at }v} \cdot \alpha_u(v) \\ &\stackrel{\eqref{eq:alpha_uv_vertex_integral},\eqref{eq:u_unmatched_integral}}{=} & \left(1-\frac{1}{2}\sum_{z<v}\marg[uz]\right) \cdot\frac{1}{2 - \sum_{z<v}\marg[uz] }  =  \frac{1}{2}.\quad\quad\quad\qedhere 
\end{eqnarray*}
\end{proof}


\paragraph{Computational aspects.}
Here we discuss how our algorithm can be implemented efficiently. Note that given the probabilities $\{x_{uv}\}_{(uv)\in E}$, we can calculate  $\{\alpha_u(v)\}_{(uv)\in E}$ by Equation~\eqref{eq:alpha_uv_vertex_integral}. Thus our batched OCRS can be implemented in polynomial time, if we are  explicitly given the probability density functions of each element in $R$. However, the batched OCRS in the reduction to prophet inequality, grants us only sample access to $R$. In a sense, it is a problem of calculating the value and estimating basic statistics of the maximum weighted matching benchmark. If we have only a sample access to $R$, we can still apply standard Monte-Carlo algorithm to estimate $x_{uv}$'s within arbitrary additive accuracy (with high probability), which leads to estimation of $\{\alpha_u(v)\}_{(uv)\in E}$ within arbitrary multiplicative accuracy (by Equation \eqref{eq:alpha_uv_vertex_integral} and the fact that $\alpha_u(v) \geq \frac{1}{2}$). This gives us a $(\frac{1}{2}-\epsilon)$-batched OCRS that runs in $\text{poly}(|V|, \frac{1}{\epsilon})$ time.


\section{A $0.337$-OCRS for Matching with Edge Arrival}
\label{sec:edge-arrival}
In this section we construct a $c$-OCRS for the edge arrival model. 
We start with a warm-up in Section \ref{sec:edge-warmup}, establishing a $\frac{1}{3}$-OCRS. In Section \ref{sec:edge-improved} we present an improved $0.337$-OCRS, using subtle observations about correlated events. 
Our results imply a prophet inequality with competitive ratio $0.337$ for the edge-arrival model. In Appendix~\ref{app:implications} we show that this guarantee holds also with respect to the optimal {\em ex-ante} matching (which is a stronger benchmark; stronger even than the optimal {\em fractional} matching).

In batched OCRS we define a sampling scheme $R$ (independent across batches), which in turn defines corresponding marginals $\margv$. 
If every batch consists of a single element (as in the model of matching with edge arrival), any vector of marginal probabilities $\margv \in [0,1]^{E}$ induces the unique sampling scheme $R$. Hence, $R$ is described by $\margv \in [0,1]^{E}$.
Based on Section~\ref{sec:matching}, for the special case of matching with edge arrivals, it suffices to construct a $c$-OCRS for every $\margv \in \polymatch$. 

Let $\margv\eqdef(\marg)_{e\in E}\in\polymatch$ be any probability vector in the (fractional) matching polytope (see Section \ref{sec:matching}).
Let $\sigma$ be an arbitrary (unknown) order of the edges.
Let $R=\sqcup_{e\in E}R_{\sigma(e)}$ be a sampling scheme that  independently generates $R_{e}$ for each edge $e$ as follows. $R_{e}=\{e\}$ with probability $\marg[e]$, and $R_{e}=\emptyset$ otherwise.
To simplify notation, we sometimes use $\marg[uv]$ to denote $\marg[(uv)]$ in $\margv\in\polymatch$.
Recall that the definition of $c$-OCRS requires the selected set $I$ to be feasible, and each element $e \in E$ to satisfy $\prob{e \in I \mid e \in R_e} \ge c$. That is, the probability that $e$ is selected given that it is in $R_e$ should be at least $c$.
Our algorithm will actually guarantee the last inequality with equality, namely that $\prob{e \in I \mid e \in R_e} = c$ for all $e\in E$. 
In the description of the algorithm and throughout this section, we write ``{\em at $(uv)$}'' or ``{\em at $e$}'' as a shorthand notation to indicate the time right {\em before} the arrival of the edge $e=(uv)$.


\begin{algorithm}
	\caption{$c$-OCRS for edge arrival}
	\label{alg:edge-integral}
	\begin{algorithmic}[1]
		\STATE At the arrival of edge $(uv)$ 
		\STATE $\quad$ Given the arrival order $\sigma_{<(uv)}$ of the edges preceding $(uv)$, calculate $\prob{u, v \mbox{ are unmatched at }(uv)}$
		\STATE $\quad$ Define 
		\begin{equation}
		\alpha_{(uv)}\eqdef \frac{c}{\prob{u,v \mbox{ are unmatched at }(uv)}}
		\label{eq:alpha_uv}
		\end{equation}
		\STATE	$\quad$ If (i) $u,v$ are unmatched, and (ii) $(uv)\in R_{(uv)}$, then match $(uv)$ with probability $\alpha_{(uv)}$.
	\end{algorithmic}
\end{algorithm}

%
%
%

Note that the term $\prob{u,v \mbox{ are unmatched at }(uv)}$ involves both randomness from $R$ and from previous steps of our algorithm.

It holds that:
\begin{equation}
\prob{(uv) \mbox{ is matched} \mid (uv) \in R_{(uv)}} = \prob{u, v \mbox{ are unmatched at }(uv)}  \cdot \alpha_{(uv)} \stackrel{\eqref{eq:alpha_uv}}{=}c,
\label{eq:pr-e-matched}
\end{equation}
which satisfies the inequality required by $c$-OCRS (Equation~\eqref{eq:cocrs_int}).


It remains to show that Algorithm \ref{alg:edge-integral} is well-defined, i.e., that $\alpha_{e} \le 1$ for all $e \in E$. 
In Section \ref{sec:edge-warmup} we show that $c=\frac{1}{3}$ can be proved using a relatively simple analysis. 
In Section \ref{sec:edge-improved} we present a more involved analysis showing that one can improve $1/3$ to $c=0.337$. 

\paragraph{Computational aspects.} The computation of $\{x_e\}_{e \in E}$ is similar to the vertex arrival setting. In fact, we can work with a stronger benchmark of the ex-ante relaxation in the edge arrival setting, which is easier from the computational view point and admits a polynomial time algorithm that finds $\{x_e\}_{e \in E}$ as the solution to the ex-ante relaxation. 
By contrast to the vertex arrival setting, given $\{x_{uv}\}_{(uv)\in E}$, it might take exponential time to precisely calculate $\{\alpha_{(uv)}\}_{(uv)\in E}$ in Algorithm~\ref{alg:edge-integral}. We still can use Monte-Carlo method to estimate $\alpha_{(uv)}$'s within arbitrary multiplicative accuracy (by the fact that $\alpha_{(uv)} \geq \frac{1}{3}$), which results in a $(0.337-\epsilon)$-OCRS that runs in $\text{poly}(|V|,\frac{1}{\epsilon})$ time.

\subsection{Warm-up: $\frac{1}{3}$-OCRS}
\label{sec:edge-warmup}
\begin{theorem}
	There is a $\frac{1}{3}$-OCRS for matching in general graphs with edge arrivals.
\end{theorem}

\begin{proof}
Let $c=\frac{1}{3}$. We prove by induction on the number of edges that all $\alpha_{e} \le 1$. The base case $|E|=1$ is trivial, since $\prob{u, v \mbox{ are unmatched at }e}=1$ and $\alpha_{e}=\frac{1}{3}$. Let us prove the induction step.  We can assume by the induction hypothesis that $\alpha_{e}\le 1$ for every edge $e\in E$ but the last arriving edge $(uv)$. To finish the induction step we need to show that $\alpha_{(uv)}\le 1$. 
Recall that our algorithm matches each edge $e$ preceding $(uv)$ with probability $c \cdot x_e$. Therefore,
\be
\label{eq:vertexMatched}
\prob{u \mbox{ is matched at } (uv)} = \sum_{s\neq v} c \cdot x_{us} \le c\quad\text{and}\quad
\prob{v \mbox{ is matched at } (uv)} = \sum_{s\neq u} c \cdot x_{sv} \le c.
\ee
Indeed, the events that $u$ is matched to the vertex $s$ for each $s\in V\setminus\{v\} $ are disjoint, $\prob{u \text{ matched to }s}=c\cdot x_{us}$, and $\sum_{s} x_{us} \le 1$; similar argument applies to $\prob{v \mbox{ is matched at } (uv)}.$  
By the union bound, we have 
\[
\prob{u,v \mbox{ are unmatched at } (uv)} \ge 1 - \prob{u \mbox{ is matched at } (uv)} - \prob{v \mbox{ is matched at } (uv)} \ge 1- 2c.
\]
For $c=1/3$, $1-2c=c$. Thus, 
\[
\prob{u,v \mbox{ are unmatched at } (uv)} \geq c\quad\text{and}\quad\alpha_{(uv)} = \frac{c}{\prob{u,v \mbox{ are unmatched at } (uv)}} \le 1,
\] 
as desired. This concludes the proof.
\end{proof}

\subsection{Improved Analysis: $0.337$-OCRS}
\label{sec:edge-improved}
In order to improve the competitive ratio beyond $1/3$, we strengthen the lower bound on the probability that $u,v$ are unmatched at $(uv)$. We again apply the same inductive argument as in the warm-up, but use more complex estimate on $\prob{u,v \mbox{ are unmatched at } (uv)}$ than a simple union bound. We denote 
\[
x_u\eqdef\sum_{s\notin\{u,v\}}  x_{us}\le 1\quad\text{and}\quad x_v\eqdef\sum_{s\notin\{u,v\}}  x_{sv}\le 1.
\]
Similar to \eqref{eq:vertexMatched} we have
\be
\label{eq:vertexMatchedRefined}
\prob{u \mbox{ is matched at } (uv)} \le c\cdot x_u\quad\text{and}\quad
\prob{v \mbox{ is matched at } (uv)}  \le c\cdot x_v.
\ee
Hence, by the inclusion-exclusion principle we have
\begin{eqnarray}
& & \prob{u,v \mbox{ are unmatched at } (uv)} \nonumber \\
& = & 1 - \prob{u \mbox{ is matched at } (uv)} - \prob{v \mbox{ is matched at } (uv)} + \prob{u,v \mbox{ are matched at } (uv)} \nonumber \\
& \ge & 1 - c \cdot (x_u + x_v) +  \prob{u,v \mbox{ are matched at } (uv)}. 
\label{eq:prob-unmatched}
\end{eqnarray}
If the matching statuses of $u$ and $v$ were independent, the bound \eqref{eq:prob-unmatched} would be $1-c(x_u + x_v) + c^2 x_u x_v \ge 1-2c+c^2$, and equating it to $c$ would yield $c \approx 0.382$. However, it is possible that the events that $u$ and $v$ are matched are negatively correlated. 
The following lemma gives a non-trivial lower bound on this correlation. 
This is the most technical lemma in this paper; its proof is the content of Section \ref{subsec:anti-correlation}. 
\begin{lemma}
	\label{lem:anti_correlation}
\[
\prob{u,v \mbox{ are umatched at } (uv)} \ge 1-2c+ \frac{c^2}{2} \cdot \left( \frac{1-2c}{1-c} \right)^2.
\]
\end{lemma}

The bound in Lemma \ref{lem:anti_correlation} leads to the construction of the improved $0.337$-OCRS.

\begin{theorem}
	\label{thm:improved-ocrs}
	There is a $0.337$-OCRS for general graphs with edge arrivals.
\end{theorem}
\begin{proof}
We set $c \approx 0.337$ to be the solution of $1-2c+ \frac{c^2}{2} \cdot \left( \frac{1-2c}{1-c} \right)^2 = c$. Then by \eqref{eq:prob-unmatched} and Lemma~\ref{lem:anti_correlation} $\prob{u,v \mbox{ are unmatched at }(uv)}\ge c$ and $\alpha_{(uv)} = \frac{c}{\prob{u,v \mbox{ are unmatched at }(uv)}} \le 1$, as required.
\end{proof}

\subsection{Proof of Lemma~\ref{lem:anti_correlation}}
\label{subsec:anti-correlation}
Fix an edge arrival order $\sigma$. We prove Lemma~\ref{lem:anti_correlation}  by induction on the number of edges (as in the warm-up in Section \ref{sec:edge-warmup}). 
The base case, where the number of edges is $|E|=1$, holds trivially. By the  induction hypothesis, we can assume that $\alpha_{e}\le 1$ for every edge $e\in E$ but the last edge $(uv)$ in $\sigma$. To simplify notations, we slightly abuse the definition of $E$ by excluding edge $(uv)$ from $E$. We need to show that $\alpha_{(uv)}\le 1$. 
By the induction hypothesis, Algorithm~\ref{alg:edge-integral} matches each edge $e\in E$ with probability exactly $c \cdot x_e$.
For the purpose of analysis, we think of the following random procedure that unifies the random realization in $R$ and the random decisions made by our algorithm.
\begin{enumerate}
	\item For each $e \in E$, $e\in R_e$ with probability $\marg[e]$, and conditioned on the event $e \in R_e$, $e$ is active with probability $\alpha_e$.
		\item Greedily pick active edges according to the arrival order $\sigma$. I.e., pick an active edge $(uv)$ if both $u$ and $v$ are unmatched at $(uv)$. 
\end{enumerate}
%
In the above procedure, each edge $e$ is \emph{active} with probability $\alpha_e \cdot \marg$ (independently across edges). 
%
Then, it is matched if both its ends are unmatched at the time the edge arrives.
In the remainder of this section we give a lower bound on the probability that both $u,v$ are unmatched at $(uv)$. 

Let $\realedges\subseteq E$ be the set of the active edges. Suppose there is a vertex $\witu$ such that $(u\witu)\in\realedges$ is the only active edge of $\witu$. Then $\witu$ must remain unmatched before $(u\witu)$. 
When $(u\witu)$ arrives, $u$ is either matched before, or it will be matched now. 
We call such $\witu$ a \emph{witness} of $u$, as existence of $\witu$ implies that $u$ is matched.
Moreover, if both $u$ and $v$ admit witnesses $\witu,\witv$, then $u, v$ must be matched at $(uv)$. Note that by definition $\witu\neq\witv$.

Let us give a lower bound on the probability that each of $u,v$ have a witness. 
We first describe a sampler $\pi_u$ that given the set of active edges $\realincv[u]\subseteq E_u\eqdef\{e\in E | e \text{ incident to }u\}$ incident to $u$, proposes a candidate witness of $u$. 
Let $\pi_u: 2^E \to V\cup \{\text{null}\}$ be the following random mapping.
\begin{enumerate}
	\item Resample each $e\in\mincv[u]\eqdef E\setminus E_u$ independently with probability $\alpha_e \cdot \marg$. Let the active edges be $\realhmincv[u]\subseteq\mincv[u]$.
	\item Run greedy on the instance $G=(V, \realincv[u]\cup \realhmincv[u])$ according to the arrival order $\sigma$.
	\item If $u$ is matched with a vertex $s$, return $\pi_u(\realincv[u])=s$; else, return null.
\end{enumerate}
The sampling procedure corresponds to the actual run of our algorithm, since $\realincv[u]\cup \realhmincv[u]$ has the same distribution as $\realedges$. Thus, the probability that $\witu$ is returned as the candidate witness of $u$ equals the probability that $(u\witu)$ is matched by our algorithm, which equals $c \cdot x_{u\witu}$. Hereafter, we denote the event that a vertex $\witu$ is chosen by the sampler $\pi_u$ as the candidate witness of a vertex $u$ by ``$\witu$ candidate of $u$". Thus

\begin{equation}
	\Prl[\realincu, \pi_u]
	{ \witu \mbox{ candidate of } u} = c \cdot \marg[u\witu]. \label{eq:prob_candidate}
\end{equation}
We also define a similar sampler $\pi_v$ to generate the candidate witness of $v$. Then,
\begin{multline}
\label{eq:witness}
\prob{u,v \mbox{ have witnesses}} \\
\ge  \sum_{\substack{\witu\ne \witv\\ \witu, \witv\notin \{u,v\} }} 
\Prl[\realedges,\pi_u,\pi_v]{\witu \text{ candidate of } u, \witv \text{ candidate of } v, |\realedges\cap E_{\witu}|=|\realedges\cap E_{\witv}|=1}\\
=  \sum_{\substack{\witu\ne \witv\\ \witu, \witv\notin \{u,v\} }}  \Prl[\realincu,\pi_u]{\witu \mbox{ candidate of } u}\cdot \Prl[\realincu,\pi_u]{(u\witv)\notin\realedges ~\Big\vert~ \witu \mbox{ candidate of } u} \\ 
\phantom{\sum_{\witu \ne \witv}} \times 
\Prl[\realincv,\pi_v]{\witv \mbox{ candidate of } v}\cdot \Prl[\realincv,\pi_v]{(v\witu)\notin\realedges ~\Big\vert~ \witv \mbox{ candidate of } v}\\
\phantom{\sum_{\witu \ne \witv}} \times 
\Prl[\realedges(G-\{u,v\})]{(s\witu) \text{ and }(s\witv) \text{ not active }\forall s \in V\setminus\{u,v\}} 
\end{multline}

\begin{lemma} 
\label{lem:conditioning}
For all $u,i,j\in V$ such that $i \ne j$,
\[	
\Prl[\realincu,\pi_u]{(ui) \mbox{ is not active }\big\vert~j \mbox{ candidate of } u} \ge \Prl[\realincu]{(ui) \mbox{ is not active}}	
\]
\end{lemma}
\begin{proof}
	As $1-\prob{(ui) \mbox{ is active }\big\vert~j \mbox{ candidate of } u}=\prob{(ui) \mbox{ is not active }\big\vert~j \mbox{ candidate of } u}$ and $1-\prob{(ui) \mbox{ is active}}=\prob{(ui) \mbox{ is not active}},$ we just need to show 
	$$\prob{(ui) \mbox{ is active}} \ge \prob{(ui) \mbox{ is active }\big\vert~j \mbox{ candidate of } u}, $$ which is  equivalent to 
	\[
	\Prl[\realincu,\pi_u]{j \mbox{ candidate of } u~\big\vert (ui) \mbox{ is active}} \le \Prl[\realincu,\pi_u]{j \mbox{ candidate of } u}. 
	\]
	%
	Two types of randomness are involved in this statement, the realization of edges $\realincu$ that are incident to $u$ and the resampling of remaining edges $\realhmincv[u]$ in $\pi_u(\realincu)$. Fix the realization of $\realincu \setminus (ui)$ and $\realhmincv[u]$. If $j$ is chosen as the candidate by $\pi_u$ when $(ui)$ is active, it must also be chosen when $(ui)$ is not active. This finishes the proof of the lemma.
\end{proof}

By Lemma~\ref{lem:conditioning}, 
$
\Prl[\realincu,\pi_u]{(u\witv)\text{ is not active } \big\vert \witu \mbox{ candidate of } u}\ge\Prl[\realincu]{(u\witv) \mbox{ is not active}}$\\ 
and 
$\Prl[\realincv,\pi_v]{(v\witu)\text{ is not active } \big\vert \witv \mbox{ candidate of } v}\ge\Prl[\realincv]{(v\witu) \mbox{ is not active}}$ in \eqref{eq:witness}. Furthermore, 
\begin{multline*}
\Prl[\realincu]{(u\witv)\notin\realedges}\times\Prl[\realincv]{(v\witu)\notin\realedges}
\times
\Prl[\realedges(G-\{u,v\})]{\forall s\ne u,v\quad(s\witu),(s\witv) \notin\realedges}
\\
= \prod_{\substack{e=(s\witv)\\ s\ne v}}\Prl[\realedges]{e\notin\realedges}
\prod_{\substack{e=(s\witu)\\ s\ne u}}\Prl[\realedges]{e\notin\realedges}
\ge 
\Prl[\widetilde{E}_{\witu}]{(\witu s) \mbox{ not active }\forall s} \cdot \Prl[\widetilde{E}_{\witv}]{(\witv s) \mbox{ not active }\forall s}.
\end{multline*}
We also know that $\prob{\witu \mbox{ candidate of } u}= c\cdot\marg[u\witu]$ and $\prob{\witv \mbox{ candidate of } v}=c\cdot\marg[v\witv]$. So we can continue the lower bound \eqref{eq:witness} on $\prob{u,v \mbox{ have witnesses}}$ as follows
\begin{align*}
\eqref{eq:witness}\ge & \sum_{\substack{\witu\ne \witv\\ \witu, \witv\notin \{u,v\} }}  c^2 \cdot \marg[u\witu] \marg[v\witv] 
\Prl[\widetilde{E}_{\witu}]{(\witu s) \mbox{ not active }\forall s} \cdot \Prl[\widetilde{E}_{\witv}]{(\witv s) \mbox{ not active }\forall s} 
\end{align*}

\begin{lemma} 	\label{lem:not-realized}
For any vertex $r$,\quad
$
\Prl[\widetilde{E}_{r}]{(rs) \mbox{ is not active }\forall s} \ge  \frac{1 - 2c}{1-c}.
$
\end{lemma}

\begin{proof}
	Without loss of generality, we assume that neighbors of $r$ are enumerated from $1$ to $k$ in such a way that among all edges incident to $r$, the edge $(ri)$ appears as the $i$-th edge in $\sigma$. Notice that each edge $(ri)$ is active independently with probability $\alpha_{ri}\marg[ri]$. Recall that $\alpha_{ri}=c\cdot\prob[\realedges]{r,i \text{ are unmatched at }(ri)}^{-1}$. As $r$ is matched to $j$ with probability $c \cdot x_{rj}$, we have by a union bound and induction hypothesis
	\begin{multline}
		\label{eq:11}
		\Prl[\realedges]{r,i \mbox{ are unmatched at } (ri)} \ge  1 - \Prl[\realedges]{i \mbox{ is matched at } (ri)} -  \Prl[\realedges]{r \mbox{ is matched at } (ri)} \\
		\ge  1 - c - \sum_{j <i} \Prl[\realedges]{(rj) \mbox{ is matched}} =  1 - c - c \sum_{j<i} \marg[rj].
	\end{multline}
	Furthermore,  each edge $(ri)$ is active independently with probability $\alpha_{ri}\marg[ri]$. Therefore, we have 
	\begin{multline*}
		\Prl[\realedges]{(ri) \mbox{ is not active }\forall  i}  =  \prod_{i=1}^k (1 - \alpha_{ri}\marg[ri]) 
		=  \prod_{i=1}^k \left( 1 - \frac{c \cdot \marg[ri]}{\Prl[\realedges]{r, i \mbox{ are unmatched at } (ri)}} \right) \\
		\ge \prod_{i=1}^k \left( 1 - \frac{c \cdot \marg[ri]}{1 - c - c \sum_{j<i} \marg[rj]} \right) 
		= \prod_{i=1}^k \frac{1 - c - c \cdot \sum_{j \le i}\marg[rj]}{1 - c - c \sum_{j<i} \marg[rj]} 
		=  \frac{1 - c - c \sum_{j\le k} \marg[rj]}{1-c} \ge  \frac{1-2c}{1-c},
	\end{multline*}
	where second equality follows by the definition of $\alpha_{ri}$, first inequality follows by Equation~\eqref{eq:11}, and the last inequality by the fact that $\sum_{j\le k} \marg[rj]\le 1$.
\end{proof}


We apply Lemma~\ref{lem:not-realized} to further simplify the lower bound of Equation \eqref{eq:witness} on $\prob{u,v \mbox{ have witnesses}}$.  
\begin{multline*}
\prob{u,v \mbox{ have witnesses}} \ge \sum_{\substack{\witu\ne \witv\\ \witu, \witv\notin \{u,v\} }}  c^2 \cdot \marg[u\witu] \marg[v\witv] \cdot \left( \frac{1-2c}{1-c} \right)^2 \\
=  c^2 \cdot \left( \frac{1-2c}{1-c} \right)^2 \cdot \left( \sum_{s\notin \{u,v\}}\marg[us] \cdot \sum_{s\notin \{u,v\}}\marg[vs] - \sum_{s\notin \{u,v\}} \marg[us] \cdot \marg[vs] \right).
\end{multline*}
We recall that $\prob{u,v \mbox{ are matched at } (uv)} \ge \prob{u,v \mbox{ have witnesses}}$. Therefore, we get the following bound from Equation~\eqref{eq:prob-unmatched}.

\begin{multline}
\label{eq:function}
 \Prl[\realedges]{u,v \mbox{ are unmatched at } (uv)} 
\ge 1 - c\InParentheses{\sum_{s\notin \{u,v\}}\marg[us] + \sum_{s\notin \{u,v\}}\marg[vs]} +  \Prl[\realedges]{u,v \mbox{ have witnesses}} \\
\ge  1 - c \InParentheses{\sum_{s}\marg[us] + \sum_{s}\marg[vs]}  + c^2 \cdot \left( \frac{1-2c}{1-c} \right)^2 \cdot  \left( \sum_{s}\marg[us] \cdot \sum_{s}\marg[vs] - \sum_{s} \marg[us] \cdot \marg[vs] \right),
\end{multline}
where all summations are taken over $s\in V\setminus\{u,v\}$. Note that since $\margv \in \polymatch$, then for all $s$ we have $\marg[us], \marg[vs]\ge 0$, $\marg[us]+\marg[vs]\le 1$, as well as $\sum_{s}\marg[us]\le 1$ and $\sum_{s}\marg[vs]\le 1$. Let us find the minimum of the function $f(\margv)\eqdef\text{RHS of}$ Equation~\eqref{eq:function}. To conclude the proof it suffices to show that $\min_{\margv  \in \polymatch} f(\margv)$ is at least the value in the statement of Lemma~\ref{lem:anti_correlation}. 

\begin{lemma}\label{lem:lastfx}
\[
f(\margv) \ge 1-2c+ \frac{c^2}{2} \cdot \left( \frac{1-2c}{1-c} \right)^2.
\]
\end{lemma}
\begin{proof}
We observe that $\frac{\partial f}{\partial \marg[ui]}=-c+c^2 \cdot \left( \frac{1-2c}{1-c} \right)^2 \cdot (\sum_{s}\marg[vs]- \marg[vi])<0$ for all $i\in V\setminus\{u,v\}$. Similarly, $\frac{\partial f}{\partial \marg[vi]}< 0$ for all $i\in V\setminus\{u,v\}$. That means that the minimum of $f$ is achieved at a boundary point $\margv$, which does not allow us to increase any of the $\marg[vi]$ or $\marg[ui]$. The analysis proceed in two cases.

\paragraph{Case (a).} If $\sum_{s}\marg[us]=\sum_{s}\marg[vs]=1$, then we can find a good upper bound on $\sum_{s} \marg[us] \cdot \marg[vs]$ as follows. First, $\frac{1}{4}\sum_{s} \marg[us] \cdot \marg[vs]\le \frac{1}{4}\sum_{s} \marg[us] \cdot (1-\marg[us])=\frac{1}{4}-\frac{1}{4}\sum_{s} \marg[us]^2$. Similarly, $\frac{1}{4}\sum_{s} \marg[us] \cdot \marg[vs]\le\frac{1}{4}-\frac{1}{4}\sum_{s} \marg[vs]^2$. Second, $\frac{1}{2}\sum_{s} \marg[us] \cdot \marg[vs]\le \frac{1}{4}\sum_{s} \marg[us]^2 +\frac{1}{4}\sum_{s} \marg[vs]^2$. Now, if we add the last three inequalities together, we get $\sum_{s} \marg[us] \cdot \marg[vs]\le \frac{1}{2}$. Thus $f(\margv)\ge 1-2c+c^2 \cdot \left( \frac{1-2c}{1-c} \right)^2\cdot(1-\frac{1}{2})$, which is equal to the desired bound in Lemma~\ref{lem:anti_correlation}.

\paragraph{Case (b).} If $\sum_{s}\marg[us]<1$, or $\sum_{s}\marg[vs]<1$. Then each inequality $\marg[us]+\marg[vs]\le 1$ must be tight for every $s\in V\setminus\{u,v\}$. It means that $\sum_{s}\marg[vs]+\sum_{s}\marg[us]=\sum_{s}(\marg[us]+\marg[vs])=\sum_{s} 1\in\Z$, also 
$\sum_{s}\marg[vs]+\sum_{s}\marg[us]<2$. Therefore, $\sum_{s}\marg[vs]+\sum_{s}\marg[us]\le 1$. We get that $f(\margv)\ge 1- c(\sum_{s}\marg[us] + \sum_{s}\marg[vs])=1-c\ge 1-2c+\frac{c^2}{2} \cdot \left( \frac{1-2c}{1-c} \right)^2$.
\end{proof}

\section{Upper Bounds for Matching Prophet Inequality}
\label{sec:edge_lower}
In this section we present upper bounds on the competitive ratios for matching prophet inequality with edge arrival, with respect to the fractional and ex-ante optimal solutions.
Note that the $1/2$ competitive ratio with respect to the classical prophet inequality extends trivially to matching prophet inequality (for both vertex and edge arrival models), implying that our $1/2$ competitive ratio for vertex arrival, as implied from Section \ref{sec:vertex-arrival}, is tight.
The following propositions give upper bounds on the competitive ratio of prophet inequalities for matching with edge arrival. Proposition~\ref{prop:upper-bound-frac} gives an upper bound with respect to the optimal fractional matching, and Proposition~\ref{prop:upper-bound-ex-ante} gives an upper bound with respect to the optimal ex-ante matching (see definition below).
\begin{proposition}
	\label{prop:upper-bound-frac}
	Under the edge arrival model, no online algorithm can get better than $\frac{3}{7}$ of $\fopt$, even for 6-vertex graphs.
\end{proposition}

\begin{proof}
Consider the graph depicted in Figure \ref{fig:upper-bounds}(a) with 6 vertices $a,b,c,d,e,f$, where edges $(ab)$, $(bc)$, $(ac)$, and $(de), (ef), (df)$ have a fixed weight of 1, and all other 9 edges have weight $\frac{1}{4\epsilon}$ with probability $\epsilon$ (for an arbitrarily small $\epsilon$), and 0 otherwise. We refer to the latter edges as the big edges. Suppose the 6 fixed edges arrive first, followed by the big edges.

The optimal fractional solution is the following: if there exists a big edge (this happens with probability $9\epsilon+O(\epsilon^2)$), then take it; else take each of the fixed edges with probability $1/2$. This approximately gives us $9\epsilon \frac{1}{4\epsilon}+(1-9\epsilon)3 = \frac{21}{4}$.

We next show that any online algorithm gets at most $\frac{9}{4}$, resulting in a ratio of $\frac{3}{7}$, as claimed.
An online algorithm can choose to select either 0, 1, or 2 fixed edges, without knowing the realization of the big edges.
If it chooses 0 fixed edges, it gets $\sim 9\epsilon \frac{1}{4\epsilon}=\frac{9}{4}$.
If it chooses 1 fixed edge, it gets $\sim 1+ 3\epsilon \frac{1}{4\epsilon}=\frac{7}{4}$.
If it chooses 2 fixed edges (one from each triangle), it gets $\sim 2+ \epsilon \frac{1}{4\epsilon}=\frac{9}{4}$.
This completes the proof.
\end{proof}

\begin{figure}[H]
	\includegraphics[width=1\textwidth]{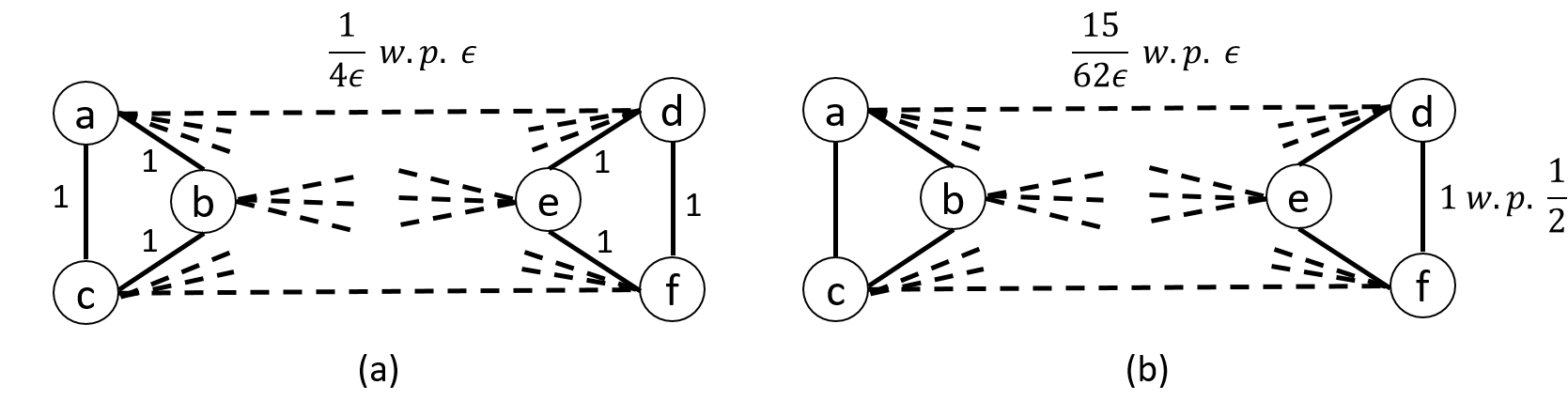}
	\caption{Upper bounds for matching prophet inequality with edge arrival. (a) upper bound with respect to optimal fractional matching. Solid lines have weight 1; dotted lines have weight $1/4\epsilon$ w.p. $\epsilon$. (b) upper bound with respect to optimal ex-ante matching. Solid lines have weight 1 w.p. $1/2$; dotted lines have weight $15/62\epsilon$ w.p. $\epsilon$.
	\label{fig:upper-bounds}}
\end{figure}

A stronger benchmark than the optimal fractional matching is the optimal ex-ante matching $\vy\in[0,1]^E$, defined as follows: 
\begin{align*}
&\vy= \argmax \sum_e \Ex[w_e]{w_e ~\vert~ w_e \geq F_e^{-1}(1-y_e)}  \cdot y_e \quad \mbox{ subject to } \vy \in \polymatch\\
&\exopt(\dists)= \sum_e \Ex[w_e]{w_e ~\vert~ w_e \geq F_e^{-1}(1-y_e)}.
\end{align*}

The following proposition gives an upper bound with respect to the optimal ex-ante matching.
Note that our lower bounds for edge arrival apply also with respect to the optimal ex-ante solution (see Section~\ref{sec:ex-ante}).

\begin{proposition}
	\label{prop:upper-bound-ex-ante}
Under the edge arrival model, no online algorithm can get better than $\frac{135}{321}$ of $\exopt$, even for 6-vertex graphs.
\end{proposition}

\begin{proof}
	Consider the graph depicted in Figure \ref{fig:upper-bounds}(b) with 6 vertices $a,b,c,d,e,f$, where edges $(ab)$, $(bc)$, $(ac)$, and $(de)$, $(ef)$, $(df)$ have a weight of 1 with probability $\frac{1}{2}$ and 0 otherwise. All other 9 edges have weight $\frac{15}{62\epsilon}$ with probability $\epsilon$ (for an arbitrarily small $\epsilon$), and 0 otherwise. We refer to the latter edges as the big edges. Suppose the edges $(ab)$, $(bc)$, $(ac)$ arrive first, followed by the edges $(de)$, $(ef)$, $(df)$, and only then the big edges arrive.
	
	The optimal \textit{ex-ante} solution is the following: it takes the big edges with probability $\epsilon$ and the other edges with probability of approximately $\frac{1}{2}$. This gives approximately a value of $9\epsilon \frac{15}{62\epsilon}+3 = \frac{321}{62}$.
	
	We next show that any online algorithm gets at most $\frac{135}{62}$, resulting in a ratio of $\frac{135}{321}$, as claimed.
	An online algorithm can choose to select either 0 or 1 edges from the set $\{(ab)$,$(bc)$,$(ca)\}$ without knowing the realization of the big edges.
	If it chooses none of the edges $(ab)$,$(bc)$,$(ca)$, it gets $\frac{1}{8}\cdot 9\epsilon\cdot \frac{15}{62\epsilon} +\frac{7}{8}\max(9\epsilon\cdot  \frac{15}{62\epsilon},1+3\epsilon\cdot\frac{15}{62\epsilon})+O(\epsilon)=\frac{135}{62}+O(\epsilon)$.
	If it chooses one edge from $\{(ab),(bc),(ca)\}$, it gets $\sim 1+ \frac{1}{8} \cdot 3\epsilon\cdot \frac{15}{62\epsilon} +\frac{7}{8}\max(3\epsilon \frac{15}{62\epsilon},1+\epsilon\frac{15}{62\epsilon} )+O(\epsilon)=\frac{135}{62}+O(\epsilon)$.
	This completes the proof.
\end{proof}

%

\section{Discussion}
\label{sec:future}
In this paper we introduce the framework of batched prophet inequalities and apply it to stochastic online matching problems. Our results demonstrate the merit of online contention resolution schemes as a useful tool for generating prophet inequalities with good performance.
The new framework introduced here suggests many fascinating avenues for future work. Some of them are listed below.
\begin{enumerate}
	\item It would be interesting to study whether our algorithms apply to online matching problems with other arrival models. For example, upon the arrival of a vertex, all edges from the new vertex to all \emph{future} vertices are revealed.
	\item We achieve an optimal $\frac{1}{2}$-competitive (resp., $0.337$-competitive) algorithm for vertex  (resp., edge) arrival via OCRS. 
	Are there pricing-based algorithms with comparable performance? 
	\item For both vertex and edge arrival settings, consider the random arrival order, a.k.a. prophet secretary. For the one-sided vertex arrivals, \citet{EhsaniHKS18} showed that the competitive ratio can be improved to $1-\frac{1}{e}$. 
	Does it generalize to two-sided vertex arrival model?
	\item 
In the edge arrival setting, it seems unlikely that an OCRS can be better than $0.382$-selectable due to the discussion in Section~\ref{sec:edge-improved}. 
	On the other hand, the best upper bound known for the prophet inequality setting is $\frac{3}{7}$ (against optimal fractional matching) and $0.420$ (against ex-ante relaxation). The gap is fairly large and it is unclear if OCRS approach can yield tight competitive ratio.
	\item This paper focuses on matching feasibility constraints. Consider studying other natural batched prophet inequality settings, with natural structures of ordered partitions into batches. 
\end{enumerate}


\bibliographystyle{plainnat}
\bibliography{prophet-matching}

\newpage
\appendix
\section{Extension to Fractional OCRS and Prophet Inequality} 
\label{sec:fractional-ocrs}
In this section we extend the definition of batched OCRS to fractional-batched OCRS, and extend the reduction in 
Theorem~\ref{thm:reduction_integral} between batched prophet inequality against the fractional optimum and fractional-batched OCRS. 
%
%
%

\subsection{Fractional Batched OCRS}
\label{app:fractional-ocrs}
Consider a {\em fractional sampling scheme} $R$ that selects a random fractional sets, i.e., vectors $\vr\in[0,1]^E$, where $\vr=(\vrt)_{t\in[T]}$ is composed of mutually independent samples $\vrt\in[0,1]^{B_t}$ (recall $E=\bigsqcup_t B_t$).
\begin{definition}[$c$-fractional-selectable batched OCRS]
An online selection algorithm  $\alg$ with respect to a fractional sampling scheme $R$ is a batched OCRS if it selects a set $I_t \subseteq B_t$ at every time $t$ such that $I \eqdef \bigsqcup_{t\in [T]} I_t\in \feasible$ is feasible.
It is a $c$-batched-OCRS if:
%
%
%
	\begin{equation}
	\Prlong[\alg]{ e \in I_t \mid \vrt=\vs} \ge c \cdot  \vs_e  \quad \text{ for all } t \in [T], \vs \in [0,1]^{B_t},  \mbox{ and } e\in B_t.\label{eq:cocrs}
	\end{equation}
	\label{def:batched_ocrs_frac}
\end{definition}

The algorithm $\alg$ is oblivious to the partition into batches and to the arrival order of the batches. It only knows the general structure $\bfamily$.
Thus, at time $t$, $\alg$ chooses $I_t$  based on $B_1, \ldots, B_t$, and $\vrt[1], \ldots, \vrt$.

\subsection{Batched Fractional Prophet Inequality} 
\label{app:fractional-pi}
For certain feasibility constraints $\feasible$, it makes sense to consider fractional optimum $\fopt(\weights)\in[0,1]^E$, where 
$\fopt\in\polymatch$ for a fractional relaxation of feasibility family $\feasible$ and $\fopt=\argmax_{\margv\in\polymatch}\iprod{\weights}{\margv}$. The weight of $\fopt$ is 
 $\weights(\fopt(\weights))\eqdef\iprod{\weights}{\fopt(\weights)}\ge\weights(\opt(\weights))$.
\begin{definition} [$c$-batched-fractional-prophet inequality]
A batched-prophet inequality algorithm $\alg$ is an online selection process that selects at time $t$ a set  $I_t \subseteq B_t$ such that $I \eqdef \bigsqcup_{t\in [T]} I_t$ is feasible (i.e., $I \in \feasible$).
We say that $\alg$ has competitive ratio $c$ against fractional optimum if
\begin{equation}
\label{eq:def-prophet-frac}
\Exlong[\weights,I]{\weights(I)} \geq c \cdot \Exlong[\weights]{\weights(\fopt(\weights))}.
\end{equation}
\end{definition}

\subsection{Reduction: Prophet Inequality to OCRS} 
\label{app:reduction-fractional}
We define the fractional random sampling scheme $R(\weights,\dists)$ for $\weights\sim\dists$ as follows.
Let $\vrt(\vwt,\dists)\eqdef \restr{\fopt(\vwt,\wweights^{(t)})}{B_t}$, where $\wweights^{(t)} \sim \dists_{-t}$ is independently generated of $\weights$, i.e., 
$\vrt(\vwt,\dists)_e = \fopt(\vwt,\wweights^{(t)})_e$ for all $e \in B_t$.
Let $\vr(\weights,\dists)\eqdef(\vrt(\vwt,\dists))_{t\in[T]}$.
%
We notice that

\begin{enumerate}
	\item The distribution of $\vr(\weights,\dists)\sim R(\weights,\dists)$ is a product distribution over the random variables $\vrt(\vwt,\dists)$.
	\item Since $\dists$ is a product distribution,  $(\vwt,\wweights^{(t)})\sim\dists$.
	\item $\forall t \in [T]$, $\vrt(\vwt,\dists)$ has the same distribution as $\restr{\fopt(\weights)}{B_t}$ (i.e., restriction to $e\in B_t$), where $\weights \sim \dists$.
	\item For every $t\in [T]$,
	\begin{equation}
		\Exlong[\weights, R]{\iprod{\vwt}{\vrt(\vwt,\dists)}} = \Exlong[\weights]{\iprod{\vwt}{\restr{\fopt(\weights)}{B_t}}}. 		\label{eq:expected_marginal_frac}
	\end{equation}
\end{enumerate}


\begin{theorem} [reduction from prophet inequality to OCRS (fractional)]
For every set $\bfamily$ of feasible ordered partitions, given a $c$-batched OCRS for the fractional sampling scheme $R(\weights,\dists)$ with $\weights\sim\dists$, one can construct a batched prophet inequality $c$-competitive algorithm for $\weights\sim\dists$ against the fractional optimum.
\end{theorem}
\begin{proof}
	Consider the following online algorithm:	
	
\vspace{0.5 cm}
	
\begin{algorithm}[H]
\caption{Reduction from $c$-batched prophet inequality to fractional $c$-batched OCRS}
\label{alg:fractional}
\begin{algorithmic}[1]
	\FOR{$t\in\{1,...,T\}$} 
	\STATE Let $\vwt$ be the weights of elements in $B_t$
	\STATE Resample the weights $\wweights^{(t)} \sim \dists_{-t}$ 
	\STATE Let $\vrt\gets \restr{\fopt(\vwt,\wweights^{(t)})}{B_t}$\quad\quad\quad\quad\quad  (i.e., $\vrt_e=\fopt(\vwt,\wweights^{(t)})_e$ for each $e\in B_t$).
	\STATE $I_t \gets$  $c$-OCRS($B_1,\ldots , B_t,\vrt[1],\ldots,\vrt$)
	\ENDFOR
	\STATE Return $I=\bigsqcup_{t \in [T]} I_t$
	\end{algorithmic}
\end{algorithm} 

\vspace{0.5 cm}

Without loss of generality, we may assume that the values of the sampling scheme are discretized, i.e., there are only countably many values in $[0,1]^{B_t}$ that $\vrt$ can take. Then
\begin{eqnarray*}
	\Exlong[\weights,R,I]{\weights(I)} & = &  
	\sum_{t \in[T]}  \Exlong[\weights,R,I]{\weights(I_t)}  \\
	 & =  & \sum_{t \in[T]} \sum_{\vs\in[0,1]^{B_t}} \Exlong[\weights,R,I]{\weights(I_t)~\Big \vert~ \vrt=\vs} \Prl[\weights,R]{\vrt=\vs}\\
	 & = & \sum_{t \in[T]} \sum_{\vs\in[0,1]^{B_t}} \Exlong[\weights,I]{\sum_{e \in B_t}w_e \cdot \Prl[I]{e \in I_t}~\bigg\vert~ \vrt=\vs} \Prl[\vwt,\vrt]{\vrt=\vs}\\
	  & \stackrel{\eqref{eq:cocrs}}{\ge} & \sum_{t \in[T]} \sum_{\vs\in[0,1]^{B_t}} \Exlong[\vwt]{\sum_{e \in B_t}w_e \cdot  c\cdot\vs_e ~\bigg\vert~ \vrt=\vs} \Prl[\vwt,\vrt]{\vrt=\vs}\\
	 & = & c \sum_{t \in[T]} \sum_{\vs\in[0,1]^{B_t}} \Exlong[\vwt]{\iprod{\vwt}{\vs} ~\big\vert~ \vrt=\vs} \Prl[\vwt,\vrt]{\vrt=\vs} \\ 
	 & = & c \sum_{t \in[T]}  \Exlong[\vwt,\vrt]{\iprod{\vwt}{\vrt} } \\ 
	 & \stackrel{\eqref{eq:expected_marginal_frac}}{=} & c \sum_{t \in[T]}  \Exlong[\weights]{\iprod{\vwt}{\restr{\fopt(\weights)}{B_t}}}
	  =  c \cdot \Exlong[\weights]{\iprod{\weights}{\fopt(\weights)}},
\end{eqnarray*}
where the third equality holds since $I_t$ and $\weights$ are independent given that $\vrt=\vs$. 
\end{proof}

\section{Stronger Benchmarks for Batched Prophet Inequality for Matching}
\label{app:implications}

In this section we show that our results for both vertex and edge arrival models hold against stronger benchmarks than the optimal integral matching. 
Specifically, for the vertex arrival model we establish guarantees against the optimal fractional matching, and for the edge arrival model, we establish guarantees against the even stronger benchmark of optimal ex-ante matching.


\subsection{Vertex Arrival: Fractional Optimum}
Let $\fopt(\weights)\eqdef\argmax_{\margmatchv\in\polymatch}\iprod{\weights}{\margmatchv}$ be the optimal fractional matching. Note that 
$$
\weights(\fopt(\weights))\eqdef\iprod{\weights}{\fopt(\weights)}\ge\weights(\opt(\weights)).
$$
In this section we show that our result for the vertex arrival model holds against the stronger benchmark of $\fopt(\weights)$, namely we can strengthen the guarantee in Definition~\ref{def:prophet} to
\[
\Exlong[\weights,I]{\weights(I)} \geq c \cdot \Exlong[\weights]{\weights(\fopt(\weights))}.
\]

Let $(B_t)_{t\in[T]}$ be a feasible ordered partition in $\bfamilyvertex$. It induces a fractional random sampling scheme $R$  with respect to the fractional optimum $\fopt$ that generates vector $\vr(\weights,\dists)\in[0,1]^E$ as defined in Section~\ref{app:reduction-fractional}. Let $\marg[uv]^{\fopt} = \Exlong{\vr_{(uv)}}$, and let $\margv^{\fopt} = (\marg[uv]^{\fopt})_{(uv) \in E}$.

%
Observe that for any edge $(uv)$, $\marg[uv]^{\fopt} = \Exlong{\fopt(\weights)_{(uv)}}$, where $\weights \sim \dists$. Therefore, $\margv^{\fopt} \in \polymatch$ (recall that $ \polymatch=\{\margmatchv ~\vert~ \forall v\in V~~ \sum_{u\in V}\margmatch[(uv)]\le 1,~\forall e\in E~~ \margmatch\ge 0\}$). We also observe that 
\begin{equation}
\sum_{u<v} \vr^{v}_{(uv)} \leq 1 \quad \mbox{for every $v \in V$ and every realization } \vr^v \in [0,1]^{B_v}, \label{eq:frac}
\end{equation}
since $\vr^{v}$ is a projection of a fractional matching on $B_v$.

With these two properties, in Appendix~\ref{sec:vertex-arrival-fractional} we construct a $\frac{1}{2}$-batched fractional OCRS for vertex arrival model, which implies a $\frac{1}{2}$-batched fractional prophet inequality for the maximum fractional matching.

\subsection{Edge Arrival: Ex-ante Optimum}
\label{sec:ex-ante}
As was previously observed by \cite{LeeS18}, for the special case in which each batch consists of a single element, one can provide the guarantees with respect to the stronger benchmark of the optimal ex-ante solution.
The optimal ex-ante solution $\vy$ is defined as follows:
$$\vy= \arg\max \sum_e \Exlong[w_e]{w_e | w_e \geq F_e^{-1}(1-y_e)}  \cdot y_e \quad \mbox{ subject to } \vy \in \polymatch.$$
Let $R^{ex-ante}(\weights,\dists)=\{e \mid w_e \geq F_e^{-1}(1-y_e)\}$. 

By definition, the distribution of $R^{ex-ante}(\weights,\dists)$ is a product distribution of $R_e^{ex-ante}(w_e,\disti[e])$. 
Therefore, any $c$-OCRS with respect to $R^{ex-ante}(\weights,\dists)$ gives us a prophet inequality algorithm with competitive ratio $c$ with respect to the optimal ex-ante solution. Specifically, our $0.337$-OCRS from Section~\ref{sec:edge-arrival} implies a $0.337$-competitive algorithm for the prophet inequality problem against the ex-ante optimum. 
Unfortunately, the reduction from general \emph{batched} prophet inequalities to \emph{batched} OCRSs does not work for ex-ante benchmark. E.g., it does not even work for the vertex arrival setting of our paper. 


\section{A $1/2$-Batched OCRS for Fractional Matching with Vertex Arrival}
\label{sec:vertex-arrival-fractional}

In what follows we extend our construction from Section \ref{sec:vertex-arrival} to a $1/2$-batched fractional OCRS for vertex arrival. Let $\vr=(\vrt[1],\ldots\vrt[|V|])$ be independent random variables over $[0,1]^{B_{i}}$.
Let $\marg[uv]=\Exlong{\vrt[v]_{(uv)}}$, and $\margv =(\marg[uv])_{(uv) \in E}$.

We write $u<v$ if vertex $u$ arrive before vertex $v$. 

\begin{theorem}
	If $\vr$ satisfies the following two conditions: 
	\begin{equation}
	 \sum_{u}\marg[uv]  \leq 1 \quad \mbox{for every } v \in V  \label{eq:sum_marg}
	\end{equation}
	\begin{equation}
	\sum_{u<v} \vr^{v}_{(uv)} \leq 1 \quad \mbox{for every $v \in V$ and every realization } \vr^v \in [0,1]^{B_v} \label{eq:sumruv}
	\end{equation} 
	Then, $\vr$ admits a $1/2$-batched fractional OCRS for the $\bfamilyvertex$ structure of batches..
\label{thm:frac-ocrs-vertex}
\end{theorem}

Note that $\vr$ as defined in Appendix \ref{app:implications} for the vertex arrival model satisfies Equations \eqref{eq:sum_marg},\eqref{eq:sumruv}.

\begin{proof}
Upon the arrival of a vertex $v$, we compute  $\alpha_{u}(v)$ for every $u<v$ as follows:
\begin{equation}
\alpha_{u}(v) \eqdef \frac{1}{2 - \sum_{z<v}\marg[uz] }\leq
\frac{1}{2 - \sum_{z}\marg[uz] }
 \stackrel{\eqref{eq:sum_marg}}{\leq} 1.
\label{eq:alpha_uv_vertex}
\end{equation}
Note that $\alpha_u(v)$ cannot be calculated before the arrival of $v$.
We claim that the following algorithm is a $\frac{1}{2}$-batched fractional OCRS with respect to $\vr$: 
\begin{algorithm}
	\caption{$1/2$-batched fractional OCRS for vertex arrival}
	\label{alg:vertex-frac}
	\begin{algorithmic}[1]
		\FOR{$v\in\{1,...,|V|\}$} 
		\STATE Calculate $\marg[uz]=\prob{(uv) \in R}$ for all $u,z<v$ and $\alpha_u(v)$ for all $u<v$.
		\STATE Among all unmatched $u<v$, choose one $u$ (or none) with probability $\vrt[v]_{(uv)}\cdot \alpha_{u}(v)$.
		\STATE If $u$ was chosen, then match $(uv)$.
		\ENDFOR
	\end{algorithmic}
\end{algorithm}

We first show that Algorithm~\ref{alg:vertex-frac} is well defined; namely, that (i) Algorithm~\ref{alg:vertex-frac} matches not more than one edge incident to $u$ and $v$, and (ii) for every $v$, we can match each available vertex $u<v$ with probability $\vr^v_{(uv)}\cdot \alpha_{u}(v)$.
The worst case is where all previous vertices are available. Thus, a sufficient condition is that $\sum_{u<v}\vr^v_{(uv)}\cdot \alpha_{u}(v) \leq 1$. 
Indeed,
$$ \sum_{u<v}\vr^v_{(uv)}\cdot \alpha_{u}(v) \stackrel{\eqref{eq:alpha_uv_vertex}}{\leq } \sum_{u<v}\vr^v_{(uv)}
\stackrel{\eqref{eq:sumruv}}{\leq} 1.$$
It remains to show that Algorithm \ref{alg:vertex-frac} is a $1/2$-batched fractional OCRS with respect to $\vr$. We fix the vertex arrival order $\sigma$. We prove by induction (on the number of vertices $|V|$) that  $\prob{(uv) \mbox{ is matched}}=\frac{\marg[uv]}{2}$. The base of induction for $|V|=0$ is trivially true. To complete the step of induction, we assume that 
$\prob{(uz) \mbox{ is matched}}=\frac{\marg[uz]}{2}$ for all $u,z<v$ and will show that $\prob{(uv) \mbox{ is matched}}=\frac{\marg[uv]}{2}$ for all $u<v$. In what follows, we say that ``$u$ is unmatched {\em at $v$}" if $u$ is unmatched {\em right before} $v$ arrives. 

%
%
%
\begin{equation}
\prob{u \mbox{ is unmatched at }v} = 1-\sum_{z<v}\prob{(uz) \mbox{ is matched}} = 1-\frac{1}{2}\sum_{z<v}\marg[uz].  \label{eq:u_unmatched} 
\end{equation}

Therefore,
\begin{eqnarray}
\prob{(uv) \mbox{ is matched}} & = & \prob{u \mbox{ is unmatched at }v}   \cdot \Exlong{\vr^v_{(uv)} \cdot \alpha_{u}(v) }\nonumber \\ 
 & \stackrel{\eqref{eq:alpha_uv_vertex},\eqref{eq:u_unmatched}}{=} &\left(1-\frac{1}{2}\sum_{z<v}\marg[uz]\right)\cdot \frac{1}{2 - \sum_{z<v}\marg[uz] } \cdot \marg[uv]\nonumber \\
& = & \frac{\marg[uv]}{2}. \nonumber
\label{eq:pr-e-matched_vertices}
\end{eqnarray}

To conclude the proof that Algorithm \ref{alg:vertex-frac} is a $\frac{1}{2}$-batched fractional OCRS with respect to $\vr$, we show that for every $u<v$ and every $\vs\in [0,1]^{B_v}$
\begin{multline*}
\prob{(uv) \in I_v \mid \vr^v =\vs}  =  \prob{u \mbox{ is unmatched at }v} \cdot  \vs_{(uv)} \cdot \alpha_u(v)  \\
 \stackrel{\eqref{eq:alpha_uv_vertex},\eqref{eq:u_unmatched}}{=} \left(1-\frac{1}{2}\sum_{z<v}\marg[uz]\right)\cdot \frac{1}{2 - \sum_{z<v}\marg[uz] } \cdot \vs_{(uv)} =  \frac{\vs_{(uv)}}{2}.\quad\quad\quad \qedhere
\end{multline*}
%
\end{proof}

\section{Other definition of batched-OCRS: bad example}
\label{sec:bad_batched_OCRS}
Here, we discuss why a natural generalization of the previous OCRS for singletons to batched-OCRS in which one simply requires that $\prob[R,I]{e \in I} \ge c \cdot x_e$ instead of Equation~\eqref{eq:cocrs_int} might be problematic. In particular, the standard reduction from a $c$-selectable OCRS to $c$-competitive prophet inequality might not work in the batched setting for such definition of $c$-selectable batched-OCRS.

Consider the following example with 4 elements $\{1,2,3,4\}$ and the downward closed family of feasible sets 
$\{\{1,3\},\{1,4\},\{2,3\},\{2,3\},\{2,4\},\{3,4\},\{1\},\{2\},\{3\},\{4\},\emptyset\}$, i.e., the maximal feasible sets have all possible 2 element 
subsets of $\{1,2,3,4\}$ except the subset $\{1,2\}$. The elements arrive in two fixed batches: $B_1=\{1,2\}$ and $B_2=\{3,4\}$.
We also consider a respective Prophet Inequality setting, in which all elements have weights independently distributed according to $\weights\sim\dists=\prod_{i=1}^4 F_i$, where
\[
F_1=F_2: \Prl[w\sim F_1]{w=\eps}=1\quad\quad\quad\text{and}\quad\quad\quad
F_3=F_4: \Prl[w\sim F_3]{w=1}=0.5,\quad \Prl[w\sim F_3]{w=0}=0.5,
\]
for some very small $\eps$. The optimum solution $\opt(\weights)$ picks the set $\{3,4\}$ if $w_3=w_4=1$, and  otherwise picks a set of size $2$ with exactly one element among $\{1,2\}$ and the larger element among $\{3,4\}$. The expected weight of the optimum is 
\[
\Exlong[\weights\sim\dists]{\weights(\opt(\weights))}=\Exlong[\weights\sim\dists]{w_3}+\Exlong[\weights\sim\dists]{w_4}+O(\eps)=1+O(\eps).
\]
The standard sampling scheme $R=R_1\sqcup R_2$ for the reduction from the Prophet inequality to OCRS observes the weights in the current batch and resample the weights of the remaining elements $\wweights^{(t)}\sim\dists_{-t}$; then it takes the set $R_t=B_t\cap\opt(\vwt,\wweights^{(t)})$ for $t\in\{1,2\}$. In our case, 
\[
R_1=\begin{cases}
\{1\} & \text{ with probability } \frac{3}{8}\\
\{2\} & \text{ with probability } \frac{3}{8}\\
\emptyset & \text{ with probability } \frac{1}{4}
\end{cases}
\quad\quad\quad
R_2=\begin{cases}
\{3\} & \text{ with probability } \frac{3}{8}\\
\{4\} & \text{ with probability } \frac{3}{8}\\
\{3,4\} & \text{ with probability } \frac{1}{4}
\end{cases}
\]
The marginal probability of the elements to be sampled in $R$ are as follows: 
\[
x_1=\prob{1\in R}=x_2=\frac{3}{8}\quad\quad\quad\quad x_3=\prob{3\in R}=x_4=\frac{5}{8}.
\]  
For such a sampling scheme $R$ one can achieve a pretty good $c$-selectable OCRS with $c=\frac{17}{20}$, by following a simple greedy algorithm that includes as many elements from $R_t$ into a feasible set $I$ as it can at each stage $t$.
In particular, this greedy OCRS would always select elements $1$ and $2$, whenever $1$, or $2$ are included in $R_1$, i.e.,
\[
\Prl[R,I]{1 \in I} = \Prl[R,I]{1 \in R}= x_1\quad\quad\quad
\Prl[R,I]{2 \in I} = \Prl[R,I]{2 \in R}= x_2.
\]
Sometimes greedy algorithm won't be able to take both $3$ and $4$ into $I$ if $R_1\neq\emptyset$, in which case it will flip a coin and take one of the $3$ or $4$ uniformly at random. Thus to calculate $\prob{3\in I}$ (similarly $\prob{4\in I}$) we consider two cases $R_1=\emptyset$ and $|R_1|=1$ and get
\begin{multline*}
\prob{3\in I}=\prob{R_1=\emptyset}\cdot\InParentheses{\vphantom{\frac{1}{2}}\prob{R_2=\{3,4\}~\vert~R_1=\emptyset}+\prob{R_2=\{3\}~\vert~R_1=\emptyset}}\\
+\prob{|R_1|=1}\cdot\InParentheses{\frac{1}{2}\prob{R_2=\{3,4\}~\vert~|R_1|=1}+\prob{R_2=\{3\}~\vert~|R_1|=1}}\\
=\frac{1}{4}\cdot\InParentheses{\frac{1}{4}+\frac{3}{8}}+\frac{3}{4}\cdot\InParentheses{\frac{1}{2}\cdot\frac{1}{4}+\frac{3}{8}}=\frac{17}{32}=\frac{17}{20}\cdot \frac{5}{8}=\frac{17}{20}\cdot x_3.
\end{multline*}
Now, if we try to convert this greedy $\frac{17}{20}$-selectable OCRS into a prophet inequality algorithm $\alg$ that selects a set $I$ with the matching marginal probabilities $\prob[I]{1\in I}=\prob[I]{2\in I}=\frac{3}{8}$, then its competitive ratio will be noticeably smaller than $\frac{17}{20}$. Indeed,
\begin{multline*}
\Exlong[\weights, I]{\alg(\weights)}=\prob{I\cap B_1=\emptyset}\cdot
\Exlong[\weights, I]{w_3+w_4~\vert~I\cap B_1=\emptyset}\\
+\prob{|I\cap B_1|=1}\cdot
\Exlong[\weights,  I]{\max(w_3,w_4)~\vert~|I\cap B_1|=1}+O(\eps)\\
=\frac{1}{4}\cdot\Exlong[\weights\sim\dists]{w_3+w_4}+\frac{3}{4}\cdot\Exlong[\weights\sim\dists]{\max(w_3,w_4)}+O(\eps)\\
=\frac{1}{4}\cdot\InParentheses{\Ex{w_3}+\Ex{w_4}}+\frac{3}{4}\cdot 1\cdot\prob{w_3=1\text{ or }w_4=1}+O(\eps)\\
=\frac{1}{4}\cdot 1 +\frac{3}{4}\cdot\frac{3}{4}+O(\eps)=\frac{13}{16}+O(\eps)=\InParentheses{\frac{13}{16}+O(\eps)}\cdot\Exlong[\weights\sim\dists]{\weights(\opt(\weights))},
\end{multline*}
i.e., the corresponding algorithm is only $\frac{13}{16}$-competitive, while we would like to have $c$-competitive algorithm with the same $c=\frac{17}{20}$ as the $c$-selectable OCRS we constructed before.

\section{Pricing Approach:  $\frac{1}{4}$ upper bound}
\label{sec:pricing_algorithm}
In this appendix we present a natural extension of the pricing-based algorithm of~\citet{FeldmanGL15} to the case of two-sided vertex arrival in bipartite matching, and show that it does not achieve a competitive ratio better than $\frac{1}{4}$. Upon arrival of a vertex $v$, the algorithm sets its price $p_v$ to be a half of the expected future contribution (to the optimum matching) of future edges incident to $v$. It then considers an edge $(uv)$ only if its weight covers the sum of the prices of its end points (i.e., $w_{uv}>p_u + p_v$). Among those, it chooses the one that maximizes $w_{uv} - p_u - p_v$. This algorithm appears as Algorithm \ref{alg:pricing} below, where $OPT(\weights)$ denotes the max-weight matching under weights $\weights$, and $u<v$ denotes that vertex $u$ arrives before vertex $v$. 
%
%
%

\begin{algorithm}
	\caption{Dynamic Pricing Algorithm\label{alg:pricing}}
	Let
	$
	p_v = \frac{1}{2}\sum_{u > v}\Ex[\weights \sim \dists]{w_{uv}\cdot \I[(uv)\in \opt(\weights)]}
	$.
	
	Let $k \in \argmax_{u<v, u \text{ unmatched}}\{w_{uv}-p_u\}$.
	
	If $w_{vk}-p_u \geq p_v$, then include $(vk)$ in the matching.
\end{algorithm}

The example depicted in Figure \ref{fig:upper-bounds14} shows that the competitive ratio of Algorithm~\ref{alg:pricing} is at most $\frac{1}{4}$. 
In this example, the expected maximum weight matching is $4-\frac{4\epsilon}{1+\epsilon}$ (by taking edge $(cd)$ if $w_{cd}>0$, and taking $(ac)$ otherwise).
Suppose the arrival order is $a,b,c,d$.
The prices calculated according to Algorithm~\ref{alg:pricing} under this arrival order are $p_a=\frac{1-\epsilon}{1+\epsilon}, p_b=0, p_c=1,p_d=0$. Given these prices, Algorithm~\ref{alg:pricing} always chooses the edge $(bc)$, which gives approximately $\frac{1}{4}$ of the expected maximum weight matching. 
\begin{figure}[H]
	\includegraphics[width=0.2\textwidth]{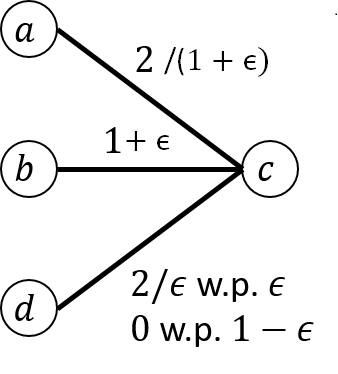}
	\caption{An upper bound of $1/4$ on the pricing-based algorithm (Algorithm~\ref{alg:pricing}) for max-weight matching with two-sided vertex arrivals.
		\label{fig:upper-bounds14}}
\end{figure}

\end{document}